\documentclass[11pt,reqno]{amsart}

\usepackage{amsmath}
\usepackage{amssymb}
\usepackage{amsfonts}
\usepackage{amsthm}
\usepackage{a4wide}
\usepackage{bbm}
\usepackage{graphics}
\usepackage{color}
\usepackage[sans]{dsfont}

\newtheorem{theorem}{Theorem}[section]
\newtheorem{lemma}[theorem]{Lemma}
\newtheorem{proposition}[theorem]{Proposition}
\newtheorem{corollary}[theorem]{Corollary}

\newtheorem{remark}[theorem]{Remark}

\numberwithin{equation}{section}

\renewcommand{\d}{\mathrm{d}}
\newcommand{\el}{\mathrm{el}}

\newcommand{\ii}{\mathrm{i}}

\def\C{{\mathbb C}}

\def\N{{\mathbb N}}
\def\R{{\mathbb R}}

\def\Z{{\mathbb Z}}

\def\<{\langle}
\def\>{\rangle}

\newcommand{\im}{\operatorname{Im}}
\newcommand{\re}{\operatorname{Re}}
\newcommand{\supp}{\operatorname{supp}}
\newcommand{\ad}{\operatorname{ad}}
\newcommand{\slim}{\mathop{\text{\rm{s-lim}}}}
\newcommand{\spec}{\operatorname{Spec}}

\newcommand{\DETAILS}[1]{}

\newcommand{\fract}[2]{\genfrac{}{}{0pt}{}{\scriptstyle #1}{\scriptstyle #2}}

\title[]{Resolvent smoothness and local decay at low energies for the standard model of non-relativistic QED}

\author[J.-F. Bony]{Jean-Fran\c{c}ois Bony}
\address[J.-F. Bony]{Institut de Math{\'e}matiques de Bordeaux \\
UMR-CNRS 5251, Universit{\'e} de Bordeaux 1 \\
351 cours de la lib{\'e}ration, 33405 Talence Cedex, France}
\email{bony@math.u-bordeaux1.fr}
\author[J. Faupin]{J{\'e}r{\'e}my Faupin}
\address[J. Faupin]{Institut de Math{\'e}matiques de Bordeaux \\
UMR-CNRS 5251, Universit{\'e} de Bordeaux 1 \\
351 cours de la lib{\'e}ration, 33405 Talence Cedex, France}
\email{jeremy.faupin@math.u-bordeaux1.fr}

\begin{document}

\begin{abstract}
We consider an atom interacting with the quantized electromagnetic field in the standard model of non-relativistic QED. The nucleus is supposed to be fixed. We prove smoothness of the resolvent and local decay of the photon dynamics for quantum states in a  spectral interval $I$ just above the ground state energy. Our results are uniform with respect to $I$. Their proofs are based on abstract Mourre's theory, a Mourre inequality established in \cite{FGS1}, Hardy-type estimates in Fock space, and a low-energy dyadic decomposition.
\end{abstract}

\maketitle

\tableofcontents

\section{Introduction and main results}

We study the dynamics of a non-relativistic atom interacting with the quantized electromagnetic field in the so-called \emph{standard model of non-relativistic QED}. If simplified to the extreme, the physical picture describing the evolution of states according to the dynamics associated with this model can be summed up as follows: As time goes to infinity, any initial state will eventually relax to the ground state by emitting photons that escape to spatial infinity. In the last years, a lot of works have been devoted to rigorous mathematical justifications of some aspects of this physical picture. In particular, among many others, we mention the following references on the proof of existence of a ground state (\cite{BFS1,GLL,BFP}), the study of resonances and lifetime of metastable states (\cite{BFS1,BFS,AFFS}),  spectral analysis (\cite{Sk,GGM,FGS1}) and partial results on scattering theory (\cite{Sp,DG,FGSc,Ge}). Completely justifying the above picture in a mathematically rigorous way would require, of course, to develop a full scattering theory for the model, in particular to prove asymptotic completeness of the wave operators, which remains an important open problem for systems of non-relativistic particles interacting with massless bosons.

In this paper, we study spectral and dynamical properties of the standard model of non-relativistic QED in the low-energy region, more precisely, in a spectral interval located just above the ground state energy and strictly below the first excited eigenvalue of the electronic Hamiltonian. In some sense, our main results will justify that the propagation velocity of low-energy photons is momentum independent, which, of course, reflects the constant speed of light. On the technical level, as is often the case, the infrared singularity intrinsic to the interaction between the atom and photons involves substantial difficulties related to the infrared problem.

Generally speaking the issue we address concerns the study of a self-adjoint operator near a threshold. The asymptotic behavior of the resolvent (and of associated quantities) near thresholds has been the subject of many studies in various fields of mathematical physics. The employed methods are varied, too. As in the work of Jensen and Kato \cite{JK}, perturbation theory can be used to consider $- \Delta + V(x)$. Resonance theory is very effective to treat dilatation analytic operators (see e.g. \cite{BFS1}) and compactly (or exponentially decaying) perturbation (see e.g. Va{\u{\i}}nberg \cite{Va}). One can also use Mourre's theory to prove limiting absorption principles at low energies. This approach was adopted, for example, by Richard \cite{Ri} who gives an abstract formalism, Bouclet \cite{Bo1,Bo2} and H\"{a}fner and the first author \cite{BH2,BH3} for long range metric perturbations of $- \Delta$, Boussaid and Gol\'enia \cite{BG} for Dirac systems, Soffer \cite{So} for $( - \Delta )^{1/2} + V(x)$, \ldots

Our paper rests on abstract results established in the framework of Mourre's theory (\cite{JMP,HSS}), in conjunction with a Mourre inequality obtained recently by Fr{\"o}hlich, Griesemer and Sigal in \cite{FGS1}. Since the work of Jensen, Mourre and Perry, \cite{JMP}, it is a well-known fact that a Mourre inequality combined with multiple commutator estimates and regularity properties yield \emph{smoothness of the resolvent}. More precisely, given a self-adjoint operator, $P$, another self-adjoint operator, $A$, conjugate to $P$ in the sense of Mourre, and a compact interval $J$ where the Mourre inequality holds, the following is satisfied:
\begin{equation}\label{d34}
\sup_{\re z \in J , \, \im z \neq 0} \left \Vert \frac{\d^n}{\d z^n} \langle A \rangle^{-n - \frac{1}{2} - \varepsilon} ( P - z )^{-1} \langle A \rangle^{-n - \frac{1}{2} - \varepsilon} \right \Vert < \infty,
\end{equation}
for any $\varepsilon>0$, where $\langle A \rangle := (1+A^2)^{1/2}$, provided that the iterated commutators $\ad^k_{A}(P)$ (defined, as usual, by $\ad^{0}_{A}(P) : = P$ and $\ad^{k+1}_{A}(P) : = [ \ad^k_{A}(P) , A ]$) are suitably bounded for $1 \leq k \leq n + 2$; (see Theorem \ref{b9} of the present paper for a precise statement). From \eqref{d34} follows the existence (and smoothness) of the boundary values of the resolvent $\langle A \rangle^{-1/2-\varepsilon} ( P - \lambda \pm \ii 0 )^{-1} \langle A \rangle^{-1/2-\varepsilon}$, and the absolute continuity of the spectrum of $P$ in $J$.

In \cite{HSS}, under similar assumptions, Hunziker, Sigal and Soffer establish the \emph{local decay} property:
\begin{equation}\label{d35}
\left \Vert \langle A \rangle^{-s } e^{-\ii t P} \chi(P) \langle A \rangle^{-s } \right \Vert \lesssim \langle t \rangle^{-s}, \quad t \in \mathbb{R},
\end{equation}
for any $s>0$ and $\chi \in \mathrm{C}_0^\infty( J )$, provided that the commutators $\ad^k_{A}(P)$ are bounded for $0 \leq k \leq n$, where $n > s + 1$ (see Theorem \ref{a21} below). It should be noted that, via Fourier transform, resolvent smoothness \eqref{d34} implies local decay \eqref{d35} with, however, the weaker rate of decay $\langle t \rangle^{-s+1/2+\varepsilon}$. Likewise, \eqref{d35} implies \eqref{d34} with the ``bigger'' weights $\langle A \rangle^{-n-1-\varepsilon}$.

For the standard model of non-relativistic QED describing an atom with static nucleus and interacting with the quantized electromagnetic field, a Mourre estimate at low-energies has been proven in \cite{FGS1}. The conjugate operator in \cite{FGS1} is the generator of dilatations in Fock space, denoted by the symbol $B$. If $\sigma \ll 1$ represents the size of the spectral interval $J_\sigma$ under consideration and its distance to the bottom of the spectrum of the Hamiltonian $H_\alpha$ (see \eqref{d25} for the definition of $H_\alpha$), then the Mourre inequality is of the form 
\begin{equation}\label{d32}
\mathds{1}_{J_\sigma}(H_\alpha) [H_\alpha , \ii B] \mathds{1}_{J_\sigma}(H_\alpha) \geq \mathrm{c}_0 \sigma \mathds{1}_{J_\sigma}(H_\alpha),
\end{equation}
for some positive constant $\mathrm{c}_0$. Assuming in addition uniform bounds with respect to $\sigma$ on the iterated commutators $\ad_B^k(\chi_\sigma(H_\alpha))$ (such bounds are proven in Appendix \ref{c12} below), this Mourre inequality yields the local decay property 
\begin{equation}\label{d33}
\left \Vert \langle B \rangle^{-s } e^{-\ii t H_\alpha} \chi_\sigma(H_\alpha) \langle B \rangle^{-s } \right \Vert \lesssim \langle \sigma t \rangle^{-s},
\end{equation}
for any $\chi_\sigma \in \mathrm{C}_0^\infty( J_\sigma )$. Similarly, one obtains bounds on weighted powers of the resolvent of the form 
\begin{equation}
\sup_{\re z \in J_\sigma , \, \im z \neq 0} \left \Vert \langle B \rangle^{-n+\frac{1}{2}-\varepsilon} ( H_\alpha - z )^{-n} \langle B \rangle^{-n+\frac{1}{2}-\varepsilon} \right \Vert \lesssim \sigma^{-n}.
\end{equation}
This non-uniformity with respect to $\sigma$ is, in fact, a typical problem one encounters when analyzing spectral and dynamical properties of a self-adjoint operator near thresholds. 

Now, it is not difficult to verify that the weights $\langle B \rangle^{-s}$ in \eqref{d33} can be replaced by $\langle X \rangle^{-s}$, where $X$ is the second quantization of the norm of the photon ``position'' operator (see \eqref{d21}). The local decay property \eqref{d33} then becomes a statement on the photon dynamics that can be interpreted as follows: Assume that the system is prepared in an initial state $\Phi$ that is in the domain of $\langle X \rangle^s$ and with spectral support in $J_\sigma$ (where, recall, $J_\sigma$ is an interval of size $\sigma$ located at a distance $\sigma \ll 1$ from the bottom of the spectrum of $H_\alpha$). Then, for large time $t \gg \sigma^{-1}$, the probability that the evolved state $e^{-\ii t H_\alpha} \Phi$ has remained in the domain of $\langle X \rangle^s$ is small (of order $\langle \sigma t \rangle^{-s}$). In other words, during the scattering process, some photons disperse to spatial infinity. This is in agreement with the physical picture mentioned above.

Our aim in this paper is the following: Since photons travel at the constant speed of light, it can be expected that resolvent smoothness and local decay hold uniformly in $\sigma$. This is precisely what we intend to prove.

Our starting point is \cite{FGS1}. For technical reasons, we consider the Mourre estimate with a modified conjugate operator, $B^\sigma$, given as the generator of dilatations in Fock space with a cutoff in the photon momentum variable; Roughly speaking, $B^\sigma$ restricts the action of $B$ to low-energy photons (see \eqref{g3} for the exact definition of $B^\sigma$). As in \eqref{d32}, the Mourre estimate is of the form $\mathds{1}_{J_\sigma}(H_\alpha) [H_\alpha , \ii B^\sigma] \mathds{1}_{J_\sigma}(H_\alpha) \geq \mathrm{c}_0 \sigma \mathds{1}_{J_\sigma}(H_\alpha)$. This inequality is established in \cite{FGS1} and is one of the main ingredients of the present paper. Next, we use methods similar to the ones of \cite{BH2,BH3}. From (second quantized versions of) Hardy's inequality, we derive bounds of the type
\begin{equation}
\left \Vert \langle X \rangle^{-s} \chi_\sigma(H_\alpha) \langle B^\sigma \rangle^s \right \Vert \lesssim \sigma^s.
\end{equation}
Thanks to a suitable low-energies dyadic decomposition, we then obtain uniform resolvent smoothness and local decay estimates, with weights expressed in terms of the second quantization of the norm of the photon position operator, $X$.

Our paper is organized as follows. Before stating our main results and comparing them with the literature in Subsection \ref{d27}, we begin with precisely defining the model we consider in Subsection \ref{d24}. In Section \ref{c20}, we recall results previously established in \cite{JMP}, \cite{HSS} and \cite{FGS1}, and we state a uniform estimate on multiple commutators. The latter is proven in Appendix \ref{c12}. In Section \ref{f4}, we derive Hardy-type estimate in Fock space as well as several other related inequalities. Our main theorems are proven in Section \ref{f5}. Finally, in Appendix \ref{d9}, various technical lemmata are gathered.

\subsection{Definition of the model} \label{d24}

We consider an atom interacting with the quantized electromagnetic field in the standard model of non-relativistic QED. The nucleus is supposed to be infinitely heavy and fixed at the origin. Moreover, to simplify the presentation, we consider a hydrogen atom, and we suppose that the electron is spinless. The Hilbert space of the total system is then the tensor product $\mathcal{H} := \mathcal{H}_{ \mathrm{el} } \otimes \mathcal{F}$, where $\mathcal{H}_{\mathrm{el} }$ is the Hilbert space for the electron given by $\mathcal{H}_{\mathrm{el} } := \mathrm{L}^2( \mathbb{R}^3 )$, and $\mathcal{F}$ is the symmetric Fock space over $\mathrm{L}^2( \mathbb{R}^3 \times \{ 1,2 \} )$, that is
\begin{equation}
\mathcal{F} := \Gamma( \mathrm{L}^2( \mathbb{R}^3 \times \{ 1,2 \} ) ).
\end{equation}
Here, for any Hilbert space $\mathfrak{h}$, $\Gamma( \mathfrak{h} )$ denotes the symmetric Fock space over $\mathfrak{h}$ defined by
\begin{equation}
\Gamma( \mathfrak{h} ) := \mathbb{C} \oplus \bigoplus_{n=1}^{\infty} \otimes^n_s \mathfrak{h},
\end{equation}
where $\otimes^n_s$ denotes the symmetric $n$th tensor product of $\mathfrak{h}$. The Hamiltonian of the model acts on $\mathcal{H}$ and is given by
\begin{align} \label{d25}
H_\alpha := \big( p + \alpha^{ \frac{3}{2} } A(\alpha x) \big)^2 + H_f +V(x),
\end{align}
where $\alpha$ is the fine-structure constant (which will be treated as a small coupling parameter), $x$ is the position of the electron, and $p := - \ii \nabla_x$. The units are chosen such that $\hbar = c = 1$. The operator $H_f = \d\Gamma( \omega )$ denotes the second quantization of the multiplication by $\omega(k) := |k|$, that is
\begin{align}
H_f := \sum_{\lambda=1,2} \int_{ \mathbb{R}^3 } |k| a^*_\lambda(k) a_\lambda(k) \d k,
\end{align}
where $a_\lambda^*(k)$ and $a_\lambda(k)$ are the usual creation and annihilation operators which obey the canonical commutation relations
\begin{align}
[ a^{\#}_\lambda(k) , a^{\#}_{\lambda'}(k') ] = 0, \quad [ a_\lambda(k) , a^*_{\lambda'}(k') ] = \delta_{\lambda\lambda'}\delta( k -k' ).
\end{align}
Here $a^{\#}$ stands for $a$ or $a^*$. For any $x \in \mathbb{R}^3$, $A(x)$ is the vector potential of the quantized electromagnetic field in the Coulomb gauge given by
\begin{align}\label{d26}
A(x) = \sum_{\lambda=1,2} \int_{\mathbb{R}^3} \frac{ \kappa(k) }{ |k|^{\frac{1}{2}} } \varepsilon_\lambda(k) \big( e^{\ii k\cdot x} a^*_\lambda(k) + e^{ - \ii k \cdot x } a_\lambda(k) \big) \d k.
\end{align}
In \eqref{d26}, the vectors $\varepsilon_\lambda(k)$, $\lambda = 1,2$, are normalized polarization vectors which are supposed to be orthogonal to each other and to $k$, and such that $\varepsilon_\lambda(k) = \varepsilon_\lambda(k/|k|)$. For instance, they can be chosen as
\begin{align}
\varepsilon_1(k) := \frac{ (-k_2 , k_1 , 0 ) }{ \sqrt{ k_1^2 + k_2^2 } }, \quad \varepsilon_2(k) := \frac{ k }{ |k| } \wedge \varepsilon_1(k).
\end{align}
Moreover, $\kappa \in \mathrm{C}_0^\infty( \mathbb{R}^3 ; \mathbb{R} )$ denotes some given ultraviolet cutoff function.
As usual, for any $f \in \mathrm{L}^2( \mathbb{R}^3 \times \{1,2\} )$, we set
\begin{align}
a^*(f) := \sum_{\lambda=1,2} \int_{\mathbb{R}^3} f(k,\lambda) a_{\lambda}^*(k) \d k, \quad a(f) := \sum_{\lambda=1,2} \int_{\mathbb{R}^3} \bar f(k,\lambda) a_{\lambda} (k) \d k,
\end{align}
and $\Phi(f) = a^*(f) + a(f)$. Hence, for any $x \in \mathbb{R}^3$, we have that
\begin{align}
A(x) = \Phi( h(x) ) \quad \text{where} \quad h(x,k,\lambda) := \frac{ \kappa(k) }{ |k|^{\frac{1}{2}} } \varepsilon_\lambda(k) e^{ \ii k \cdot x }.
\end{align}
The external potential $V$ belongs to $\mathrm{L}^2_{\mathrm{loc}}( \mathbb{R}^3 )$ and is supposed to be $\Delta$-bounded with relative bound 0. We assume in addition that $e_1 := \inf \spec ( -\Delta + V )$ is a simple isolated eigenvalue. We set $e_2 := \inf ( \spec( -\Delta + V ) \setminus \{ e_1 \} )$ and $e_{\mathrm{gap}} := e_2 - e_1 > 0$.

\subsection{Main results} \label{d27}

Let $E_\alpha := \inf \spec ( H_\alpha )$. It follows from \cite{BFS,GLL} that $E_\alpha$ is an eigenvalue of $H_\alpha$. Let $\Pi_\alpha$ be the projection onto the eigenspace associated with $E_\alpha$, and let $\bar \Pi_\alpha := \mathds{1} - \Pi_\alpha$. Thus, in particular, $\Pi_{0} = \pi_{1} \otimes \Pi_{\Omega}$, where $\pi_1$ denotes the projection onto the eigenspace of $H_\el$ associated with $e_1$, and $\Pi_\Omega$ is the orthogonal projection onto the Fock vacuum $\Omega : = ( 1 , 0 , 0 , \dots )$. Likewise, $\bar \Pi_0 = \bar \pi_1 \otimes \mathds{1} + \pi_1 \otimes \bar \Pi_\Omega$, with $\bar \pi_1 = \mathds{1} - \pi_1$ and $\bar \Pi_\Omega = \mathds{1} - \Pi_\Omega$. To simplify notations, let
\begin{equation} \label{d21}
X := \d \Gamma ( | \ii \nabla_k | ) ,
\end{equation}
denote the second quantization of the norm of the photon position operator. Our main results are stated in Theorems \ref{d28}, \ref{b8} and \ref{c1}.

\begin{theorem}[Limiting absorption principle]\sl \label{d28}
There exists $\alpha_c > 0$ such that, for all $s > 1/2$, there exists $\mathrm{C}_s>0$ such that, for all $0 \leq \alpha \leq \alpha_c$,
\begin{equation*}
\sup_{z \in \mathbb{C} \setminus \mathbb{R}, \, \re z \leq E_\alpha + e_{\mathrm{gap}}/4} \big\Vert \langle X \rangle^{-s} \big( H_\alpha - z \big)^{-1} \bar \Pi_\alpha \langle X \rangle^{-s} \big\Vert \leq \mathrm{C}_s .
\end{equation*}
\end{theorem}

\begin{remark}\sl
$i)$ The fact that the assumption $s > 1/2$ is sufficient for the limiting absorption principle to hold illustrates that the propagation velocity of photons does not depend on their momentum.

$ii)$ Theorem \ref{d28} implies that the spectrum of $H_\alpha$ in $(E_\alpha,E_\alpha+e_{\mathrm{gap}}/4)$ is purely absolutely continuous (see e.g. \cite[Theorem XIII.20]{RS4}), which was already proven in \cite{FGS1}.
\end{remark}

The next result provides a further information on the regularity of the weighted resolvent.

\begin{theorem}[Resolvent smoothness]\sl \label{b8}
There exists $\alpha_c > 0$ such that, for all $1/2 < s < 3/2$ and $\varepsilon >0$, there exists $\mathrm{C}_{s , \varepsilon} >0$ such that, for all $0 \leq \alpha \leq \alpha_c$,
\begin{equation*}
\big\Vert \< X \>^{- s} \big( ( H_\alpha - z )^{-1} - ( H_\alpha - z^{\prime} )^{-1} \big) \bar{\Pi}_\alpha \< X \>^{- s} \big\Vert \leq \mathrm{C}_{s, \varepsilon} \vert z - z^{\prime} \vert^{s - \frac{1}{2} - \varepsilon} ,
\end{equation*}
uniformly for $z,z' \in \C \setminus \R$ with $\re z, \re z^{\prime} \leq E_{\alpha} + e_{ \mathrm{gap} } / 4$ and $\im z \cdot \im z^{\prime} > 0$.
\end{theorem}

Theorem \ref{b8} and a standard argument imply that the weighted resolvent has limits on the real axis. More precisely, letting $\mathcal{B}(\mathcal{H})$ denote the set of bounded operators in $\mathcal{H}$, we have

\begin{corollary}\sl \label{c2}
For all $0 \leq \alpha \leq \alpha_{c}$, $s > 1/2$ and $\lambda \leq E_\alpha + e_{\mathrm{gap}} / 4$, the limits
\begin{equation*}
\langle X \rangle^{- s} ( H_\alpha - \lambda \pm \ii 0 )^{-1} \bar{\Pi}_\alpha \langle X \rangle^{-s} : = \lim_{\mu \downarrow 0} \langle X \rangle^{-s} ( H_\alpha - \lambda \pm \ii \mu )^{-1} \bar{\Pi}_\alpha \langle X \rangle^{-s}
\end{equation*}
exist in the norm topology of $\mathcal{B}( \mathcal{H} )$. Moreover, for $1/2 < s < 3/2$ and $\varepsilon >0$, the maps
\begin{equation*}
(-\infty , E_\alpha + e_{\mathrm{gap}}/4 ] \ni \lambda \longmapsto \langle X \rangle^{-s} ( H_\alpha - \lambda \pm \ii 0 )^{-1} \bar \Pi_\alpha \langle X \rangle^{-s} \in \mathcal{B} ( \mathcal{H} )
\end{equation*}
are H\"{o}lder continuous of order $s - 1/2 - \varepsilon$ with respect to $\lambda$. 
\end{corollary}

Eventually, we prove the following local decay property.

\begin{theorem}[Local decay]\sl \label{c1}
There exists $\alpha_c > 0$ such that, for all $\chi \in \mathrm{C}_0^\infty( ( -\infty ,E_\alpha + e_{\mathrm{gap}}/4 ) ; \mathbb{R} )$ and $0 \leq s < 2$, we have
\begin{equation*}
\langle X \rangle^{- s} e^{- \ii t H_{\alpha}} \chi ( H_{\alpha} ) \langle X \rangle^{- s} = e^{- \ii t E_{\alpha}} \chi ( E_{\alpha} ) \langle X \rangle^{- s} \Pi_{\alpha} \langle X \rangle^{- s} + {\mathcal O} ( \< t \>^{- s} ) ,
\end{equation*}
for all $t \in \R$, uniformly with respect to $0 \leq \alpha \leq \alpha_c$.
\end{theorem}

In the previous statement, ${\mathcal O} ( \< t \>^{- s} )$ stands for an operator bounded by $\mathrm{C} \< t \>^{-s}$ where $\mathrm{C}$ is uniform in $t \in \R$ and $0 \leq \alpha \leq \alpha_c$.

\begin{remark}\sl
$i)$ By Fourier transform, Corollary \ref{c2} implies the local decay property
\begin{equation*}
\langle X \rangle^{- s} e^{- \ii t H_{\alpha}} \chi ( H_{\alpha} ) \langle X \rangle^{- s} = e^{- \ii t E_{\alpha}} \chi ( E_{\alpha} ) \langle X \rangle^{- s} \Pi_{\alpha} \langle X \rangle^{- s} + {\mathcal O} \big( \< t \>^{- s + \frac{1}{2} + \varepsilon} \big) ,
\end{equation*}
for $1/2 < s < 3/2$ and $\varepsilon >0$, which is of course weaker than Theorem \ref{c1}.

$ii)$ The restrictions $s < 3/2$ in Theorem \ref{b8} and $s < 2$ in Theorem \ref{c1} are due to the infrared singularity of the model. More precisely, if one replaces the electromagnetic vector potential in \eqref{d26} by its infrared regularized version,
\begin{align*}
A_\mu(x) := \sum_{\lambda=1,2} \int_{\mathbb{R}^3} \frac{ \kappa(k) }{ |k|^{\frac{1}{2}-\mu} } \varepsilon_\lambda(k) \big( e^{\ii k\cdot x} a^*_\lambda(k) + e^{ - \ii k \cdot x } a_\lambda(k) \big) \d k,
\end{align*}
for $0 \le \mu \le 1$, then one can verify that Theorem \ref{b8} holds for any $s < 3/2 + \mu$ and Theorem \ref{c1} for any $ s < 2 + \mu$ (provided in addition that the weights $\langle X \rangle^{-s}$ are replaced by the bigger ones $\langle \d \Gamma ( \langle \ii \nabla_k \rangle )  \rangle^{-s}$). For $\mu >1$, the proofs we give do not yield better results than $s<5/2$ in Theorem \ref{b8} and $s<3$ in Theorem \ref{c1}. Here the restrictions are due to Hardy's inequality $\| |k|^{-s} \varphi \| \le \mathrm{C}_s \| | \ii \nabla_k |^s \varphi \|$ in $\mathrm{L}^2( \mathbb{R}^3)$, which is valid provided that $0 \le s < 3/2$.

$iii)$ One could replace $e_{\mathrm{gap}}/4$ by $e_{\mathrm{gap}} - \delta$, with $\delta>0$, in the statements of Theorems \ref{d28}, \ref{b8} and \ref{c1}. Of course the critical value $\alpha_c$ would then depend on $\delta$.
\end{remark}

As mentioned in the introduction, our main achievement compared to \cite{FGS1} lies into the fact that our results do not depend on the spectral interval $I \subset (E_\alpha , E_\alpha + e_{\mathrm{gap}}/4)$ on which resolvent smoothness and local decay are proven. Recently, several papers have been devoted to spectral analysis at low-energies for various quantum field theory models. We mention \cite{FGS2} where, for the standard model of non-relativistic QED, an alternative proof of the limiting absorption principle is given, based on an application of the spectral renormalization group, \cite{CFFS} where a dressed electron in non-relativistic QED is considered, and \cite{ABFG} where a mathematical model of the weak interaction is studied. In all the previously cited papers, however, the obtained estimates are not uniform with respect to the considered spectral interval.

Let us also mention that another approach has been used in the literature to study spectral and dynamical properties of non-relativistic, massless quantum field theory models. Instead of the generator of dilatations, one can, in some cases, consider the generator of radial translations, say $\widetilde{B}$, as a conjugate operator (see e.g. \cite{Ge,GGM}). Since $\widetilde{B}$ is not self-adjoint and since the commutator of $\widetilde{B}$ with the considered Hamiltonian $H$ cannot be controlled by any powers of the resolvent of $H$, serious technical difficulties appear to implement the Mourre method. Nevertheless, as shown by Georgescu, G\'erard and M{\o}ller, Mourre's theory can be extended to cover such a case. Within this approach, it may be possible to obtain a uniform Mourre estimate (at least for the Nelson model that is considered in \cite{GGM}) and, hence, uniform bounds on the resolvent and on local decay. Indeed, in some sense, if the generator of radial translation is chosen as the conjugate operator, the bottom of the spectrum is not a threshold anymore. Another significant advantage of the approach of \cite{GGM} is that the obtained results hold for any value of the coupling constant; It is presently not known whether similar results can be proven using the generator of dilatations instead. On the other hand, the infrared singularity in the iterated commutators $\ad^j_{\widetilde{B}}(H)$ is increased by a power $|k|^{-j}$ (while the order of the singularity does not change when commuting with the generator of dilatations), which makes difficult to control these iterated commutators unless one imposes from the beginning some regularity assumption on the form factor Hamiltonian.

\section{Preliminary results} \label{c20}

\subsection{Abstract setting}\label{c13}

Let $\mathcal{H}$ be a separable Hilbert space and let $P$, $A$ be two self-adjoint operators on $\mathcal{H}$. We recall that, if $P$ is bounded, $P$ is said to be in $\mathrm{C}^n(A)$ if and only if the map
\begin{align}
s \mapsto e^{- \ii s A} P e^{\ii s A} \Phi,
\end{align}
is of class $\mathrm{C}^n(\mathbb{R})$ for all $\Phi \in \mathcal{H}$. This property is equivalent to the fact that, for $1 \leq k \leq n$, the commutators $\ad_A^k(P)$ defined as quadratic forms on $D(A) \times D(A)$ extend by continuity to bounded quadratic forms on $\mathcal{H} \times \mathcal{H}$. If $P$ is unbounded, $P$ is said to be in $\mathrm{C}^n(A)$ if and only if $(P-z)^{-1}$ is in $\mathrm{C}^n(A)$ for some (and hence for all) $z$ in the resolvent set of $P$.

We recall that if $P$ is in $\mathrm{C}^1(A)$, then $D(P) \cap D(A)$ is a core for $P$ and the quadratic form $[P,A]$ defined on $( D(P) \cap D(A) ) \times ( D(P) \cap D(A) )$ extends by continuity to a bounded quadratic form on $D(P) \times D(P)$ (denoted by the same symbol). Moreover, if $P$ is in $\mathrm{C}^1(A)$, then $(P-z)^{-1}$ preserves $D(A)$ for all $z$ in the resolvent set of $P$. We also recall that if $P$ is in $\mathrm{C}^n(A)$, then for all $\varphi \in \mathrm{C}_0^\infty( \mathbb{R} ; \mathbb{R} )$, $\varphi(P)$ is in $\mathrm{C}^n(A)$.

An operator $P$ in $\mathrm{C}^1(A)$ is said to satisfy a Mourre estimate with respect to $A$ on a bounded open interval $I \subset \spec (P)$ if there exists a positive constant $\mathrm{c}_0$ such that
\begin{align}\label{d29}
\mathds{1}_I(P) [P , \ii A] \mathds{1}_I(P) \geq \mathrm{c}_0 \mathds{1}_I(P),
\end{align}
in the sense of quadratic forms on $\mathcal{H} \times \mathcal{H}$.

We now state a standard result on the power of the resolvent \cite{JMP}.

\begin{theorem}[Jensen, Mourre, Perry]\sl \label{b9}
For $n \in \N \cup \{0\}$, let $P , A$ be two self-adjoint operators such that $P \in \mathrm{C}^{n + 1} (A)$, that the commutators $\ad^{j}_{\ii A} P$ are bounded for $1 \leq j \leq n + 1$ and that the Mourre estimate \eqref{d29} holds with $c_{0} >0$ and $I \subset \spec (P)$ an open interval. Then, for all $J \Subset I$ and $\varepsilon>0$, there exists a positive constant $\mathrm{C}_{J,\varepsilon}$ such that
\begin{equation}\label{d31}
\sup_{\fract{\re z \in J}{\im z \neq 0}} \big\Vert \< A \>^{- n + \frac{1}{2} - \varepsilon} ( P - \lambda )^{-n} \< A \>^{- n + \frac{1}{2} - \varepsilon} \big\Vert \leq \mathrm{C}_{J,\varepsilon} .
\end{equation}
\end{theorem}

The following abstract result is taken from \cite{HSS}.

\begin{theorem}[Hunziker, Sigal, Soffer]\sl \label{a21}
Let $s >0$ and $\overline{s} = \min \{ n \in \mathbb{N} ; \, n > s +1 \}$. Let $P , A$ be two self-adjoint operators such that $P$ is bounded, $P \in \mathrm{C}^{\overline{s}} (A)$, and the Mourre estimate \eqref{d29} holds with $\mathrm{c}_{0} >0$ and $I \subset \spec (P)$ an open interval. Then, for all $\chi \in \mathrm{C}_0^{\infty} ( I )$, there exists a positive constant $\mathrm{C}_{\chi,s}$ such that
\begin{equation} \label{d30}
\big\Vert \langle A \rangle^{- s} e^{-i t P} \chi (P) \langle A \rangle^{-s} \big\Vert \leq \mathrm{C}_{\chi,s} \langle t \rangle^{- s}.
\end{equation}
\end{theorem}

Moreover, from the proofs of these results, the constants $\mathrm{C}_{J,\varepsilon}$ and $\mathrm{C}_{\chi,s}$ appearing in \eqref{d31} and \eqref{d30} only depend on the constant $\mathrm{c}_{0}$ and on $\Vert \ad^{j}_{\ii A} P \Vert$. In other words, if $P$ and $A$ depend on a parameter in a such way that the Mourre estimate and the upper bounds on the commutators are uniform with respect to this parameter, then the constants $\mathrm{C}_{J,\varepsilon}$ and $\mathrm{C}_{\chi,s}$ in the conclusion of Theorem \ref{b9} and Theorem \ref{a21} do not depend on the parameter.

\subsection{Infrared decomposition}\label{d36}

In this subsection, we introduce notations related to the infrared decomposition which will be an important tool in our proof of Theorem \ref{d28}, Theorem \ref{b8} and Theorem \ref{c1}.

For any $\sigma>0$, let $\mathcal{F}^{\le\sigma}$ and $\mathcal{F}_{\ge\sigma}$ denote the Fock spaces over $\mathrm{L}^2( \{ (k,\lambda) , |k| \leq \sigma \} )$ and $\mathrm{L}^2( \{ (k,\lambda) , |k| \geq \sigma \} )$ respectively. The Hilbert spaces $\mathcal{F}$ and $\mathcal{F}_{\ge\sigma} \otimes \mathcal{F}^{\le\sigma}$ are isomorphic and we shall not distinguish between the two of them. Moreover we set $\mathcal{H}_{\ge\sigma} := \mathrm{L}^2( \mathbb{R}^3 ) \otimes \mathcal{F}_{\ge\sigma}$. The infrared cutoff Hamiltonian acts on $\mathcal{H}$ and is defined by
\begin{align}
H_{\alpha,\sigma} := \big( p + \alpha^{ \frac{3}{2} } A_{\ge\sigma}(\alpha x) \big)^2 + H_f +V(x),
\end{align}
where
\begin{align}\label{f3}
A_{\ge\sigma}(x) := \Phi( h_{\ge\sigma}(x) ), \quad h_{\ge\sigma}(x,k,\lambda) := \mathds{1}_{|k|\ge\sigma}(k) h(x,k,\lambda).
\end{align}
Using Lemma \ref{d10} and the Kato-Rellich theorem, it is not difficult to verify that for all $\sigma\ge0$ and $\alpha$ small enough, $H_{\alpha,\sigma}$ is self-adjoint with domain $D(H_{\alpha,\sigma}) = D(H_0)$. The restriction of $H_{\alpha,\sigma}$ to $\mathcal{H}_{\ge\sigma}$ is denoted by $K_{\alpha,\ge\sigma}$, so that we can decompose
\begin{align}
H_{\alpha,\sigma} = K_{\alpha,\ge\sigma} \otimes \mathds{1}_{ \mathcal{F}^{\le\sigma}}+ \mathds{1}_{ \mathcal{H}_{\ge\sigma} } \otimes H_f.
\end{align}
Let $E_{\alpha,\sigma} := \inf \spec ( H_{\alpha,\sigma} ) = \inf \spec ( K_{\alpha,\ge\sigma} )$ and let $\Pi_{\alpha,\ge\sigma} := \mathds{1}_{ \{ E_{\alpha,\sigma}\} }( K_{\alpha,\ge\sigma} )$. We recall the following proposition from \cite{BFP,FGS1}:

\begin{proposition}\sl \label{g1}
There exists $\alpha_c>0$ such that, for all $0\le\alpha\le\alpha_c$ and $0 < \sigma \leq e_{\mathrm{gap}} / 2$,
\begin{equation*}
\spec ( K_{\alpha,\ge\sigma} ) \cap ( E_{\alpha,\sigma} , E_{\alpha,\sigma} + \sigma ) = \emptyset.
\end{equation*}
\end{proposition}

Let us mention that, in \cite{BFP}, Proposition \ref{g1} is proven for some sequence $\sigma_n \to 0$, while in \cite{FGS1}, the result is established with a smooth infrared cutoff. However, slightly modifying the proof of \cite{BFP}, it is not difficult to obtain the gap property as stated in Proposition \ref{g1}.

For any $\varphi \in \mathrm{C}_0^\infty( (0,1) )$, we set $\varphi_\sigma(\cdot) := \varphi( \cdot / \sigma )$. Observe that by Proposition \ref{g1}, for any $\varphi \in \mathrm{C}_0^\infty( (0,1) ; \mathbb{R} )$ and $0 < \sigma \leq e_{ \mathrm{gap} } / 2$, we have
\begin{equation}\label{g2}
\varphi_\sigma( H_{\alpha,\sigma} - E_{\alpha,\sigma} ) = \Pi_{\alpha,\ge\sigma} \otimes \varphi_\sigma( H_f ).
\end{equation}

Let $\eta_{\sigma} ( k ) : = \eta ( k / \sigma )$ with $\eta \in \mathrm{C}_0^\infty( \mathbb{R}^3 ; [0,1] )$ such that $\eta(k) = 1$ if $| k | \leq 1/2$ and $\eta ( k ) = 0$ if $| k | \geq 1$. The generator of dilatations in Fock space with a cutoff in the momentum variable is denoted by $B^\sigma$ and is defined by
\begin{equation} \label{g3}
B^\sigma := \d \Gamma ( b^\sigma ), \quad b^\sigma :=  \frac{\ii}{2} \eta_\sigma(k) ( k \cdot \nabla_k + \nabla_k \cdot k ) \eta_\sigma(k).
\end{equation}
Finally we set
\begin{equation} \label{g4}
A^{\le\sigma}(x) := \Phi( h^{\le\sigma}(x) ), \quad h^{\le\sigma}(x,k,\lambda) := \mathds{1}_{|k|\le\sigma}(k) h(x,k,\lambda).
\end{equation}

We recall the following commutation properties which  will be used in the sequel:
\begin{align}
& [ A(x) , \ii B^\sigma ] = [ A^{\le\sigma}(x) , \ii B^\sigma ] = - \Phi( \ii b^\sigma h^{\le\sigma}(x) ), \label{g5} \\
& [ H_f , \ii B^\sigma ] = \d \Gamma ( \eta_\sigma(k) ^2 |k| ), \label{g6}
\end{align}
in the sense of quadratic forms on $D( H_0 ) \cap D( B^\sigma )$.

\subsection{The Mourre inequality and multiple commutators estimates}

In this section, we recall the Mourre estimate obtained in \cite{FGS1} with $B^\sigma$ as a conjugate operator. We also state uniform estimates on the commutators of $B^\sigma$ with functions of $H_\alpha$. For the convenience of the reader, the proof of these multiple commutators estimates are deferred to Appendix \ref{c12}. Applying the abstract results from Subsection \ref{c13}, we then deduce resolvent smoothness and local decay estimates for $H_\alpha$, with weights expressed in terms of $B^\sigma$. Notice that, here, the obtained estimates are not uniform in $\sigma$.

The next Mourre estimate follows from Proposition 6, Proposition 7, Lemma 17 of \cite{FGS1} and their proof.

\begin{theorem}[Fr\"ohlich, Griesemer, Sigal]\sl \label{c17}
Let $I \Subset ( 0 , 1 )$ be an open interval. There exist $\alpha_c>0$ and $\mathrm{c}_{0} > 0$ such that, for all $0 \leq \alpha \leq \alpha_c$ and $0 < \sigma \leq e_{\mathrm{gap}} / 2$,
\begin{equation*}
\mathds{1}_{\sigma I} ( H_{\alpha} - E_{\alpha} ) [ H_{\alpha} , \ii B^{\sigma} ] \mathds{1}_{\sigma I} ( H_{\alpha} - E_{\alpha} ) \geq \mathrm{c}_{0} \sigma \mathds{1}_{\sigma I} ( H_{\alpha} - E_{\alpha} ) .
\end{equation*}
\end{theorem}

The following lemma is proven in Appendix \ref{c12} below.

\begin{lemma}\sl \label{c18}
There exists $\alpha_c>0$ such that, for all $n \in \mathbb{N}\cup\{0\}$ and $\varphi \in \mathrm{C}_{0}^{\infty} ( ( - \infty , 1 ) ; \mathbb{R} )$, there exists a positive constant $\mathrm{C}_{n,\varphi}$ such that, for all $0 \leq \alpha \leq \alpha_c$ and $0 < \sigma \leq e_{\mathrm{gap}} / 2$,
\begin{equation*}
\big\Vert \ad^n_{\ii B^\sigma}( \varphi_\sigma ( H_\alpha - E_\alpha ) ) \big\Vert \leq \mathrm{C}_{n,\varphi}.
\end{equation*}
\end{lemma}

Combining these results with Theorem \ref{b9}, we get the following proposition.

\begin{proposition}\sl \label{c8}
Let $I \Subset ( 0 , 1 )$ be a open interval. There exists $\alpha_c>0$ such that, for all $n \in \N$ and $\varepsilon > 0$, there exists $\mathrm{C}_{n , \varepsilon} > 0$ such that, for all $0 \leq \alpha \leq \alpha_c$ and $0 < \sigma \leq e_{\mathrm{gap}} / 2$,
\begin{equation*}
\sup_{ z \in \mathbb{C} \setminus \mathbb{R}, \, \re z \in E_\alpha + \sigma I} \big\Vert \langle B^\sigma \rangle^{-n + \frac{1}{2} - \varepsilon} ( H_\alpha - z )^{-n} \langle B^\sigma \rangle^{-n + \frac{1}{2}-\varepsilon} \big\Vert \leq \frac{\mathrm{C}_{n , \varepsilon}}{ \sigma^n }.
\end{equation*}
\end{proposition}

\begin{proof}
Let $J$ be an open interval such that $I \Subset J \Subset ( 0 , 1 )$ and let $\psi \in \mathrm{C}^{\infty} ( \mathbb{R} ; [ 0 , + \infty ) )$ be a non-decreasing function such that
\begin{align*}
\psi (x) = \left\{ \begin{aligned}
&0 &&\text{near } ( - \infty , 0] ,  \\
&x &&\text{near } J ,  \\
&\text{Const.} &&\text{near } [ 1 , + \infty ) .
\end{aligned} \right.
\end{align*}
In particular, $\psi$ is a bijection from $J$ onto itself. We define
\begin{align*}
P_{\alpha , \sigma} : = \psi_{\sigma} ( H_{\alpha} - E_{\alpha} ) .
\end{align*}
From the properties of $\psi$, Theorem \ref{c17} yields
\begin{align}
\mathds{1}_{J} ( P_{\alpha , \sigma} ) [ P_{\alpha , \sigma} , \ii B^{\sigma} ] \mathds{1}_{J} ( P_{\alpha , \sigma}) &= \mathds{1}_{\sigma J} ( H_{\alpha} - E_{\alpha} ) [ H_{\alpha} , \ii B^{\sigma} ] \mathds{1}_{\sigma J} ( H_{\alpha} - E_{\alpha} )   \nonumber \\
&\geq \mathrm{c}_{0} \mathds{1}_{\sigma J} ( H_{\alpha} - E_{\alpha} ) = \mathrm{c}_{0} \mathds{1}_{J} ( P_{\alpha , \sigma} ) ,   \label{a19}
\end{align}
uniformly for $0 < \sigma \leq e_{\mathrm{gap}} / 2$ and $0 \leq \alpha \leq \alpha_c$.

On the other hand, since $H_{\alpha} - E_{\alpha} \geq 0$, we can write
\begin{align*}
P_{\alpha , \sigma} = \psi_{\sigma} ( H_{\alpha} - E_{\alpha} ) = \text{Const.} + \widetilde{\psi}_{\sigma} ( H_{\alpha} - E_{\alpha} ) ,
\end{align*}
for some $\widetilde{\psi} \in \mathrm{C}_0^\infty ( ( - \infty , 1 ) )$. Then, Lemma \ref{c18} implies that, for all $j \in \mathbb{N} \cup \{ 0 \}$, there exists a positive constant $\mathrm{C}_{j}$ such that
\begin{align}\label{a20}
\big\Vert \ad^j_{\ii B^\sigma} ( P_{\alpha , \sigma} ) \big\Vert \leq \mathrm{C}_{j} ,
\end{align}
uniformly for $0 < \sigma \leq e_{\mathrm{gap}} / 2$ and $0 \leq \alpha \leq \alpha_c$. 

Therefore, combining the uniform Mourre estimate \eqref{a19} and the uniform upper bounds \eqref{a20} with Theorem \ref{b9} and the remark below Theorem \ref{a21}, we get
\begin{equation*}
\sup_{w \in \mathbb{C} \setminus \mathbb{R}, \, \re w \in I} \big\Vert \langle B^\sigma \rangle^{-n + \frac{1}{2} - \varepsilon} ( P_{\alpha , \sigma} - w )^{-n} \langle B^\sigma \rangle^{-n + \frac{1}{2}-\varepsilon} \big\Vert \leq \mathrm{C}_{n , \varepsilon} ,
\end{equation*}
and then
\begin{equation*}
\sup_{z \in \mathbb{C} \setminus \mathbb{R}, \, \re z \in E_\alpha + \sigma I} \big\Vert \langle B^\sigma \rangle^{-n + \frac{1}{2} - \varepsilon} \big( ( \sigma P_{\alpha , \sigma} + E_{\alpha} ) - z \big)^{-n} \langle B^\sigma \rangle^{-n + \frac{1}{2}-\varepsilon} \big\Vert \leq \frac{\mathrm{C}_{n , \varepsilon}}{ \sigma^n } ,
\end{equation*}
uniformly for $0 < \sigma \leq e_{\mathrm{gap}} / 2$ and $0 \leq \alpha \leq \alpha_c$. Now the proposition follows from
\begin{equation*}
\sup_{z \in \mathbb{C} \setminus \mathbb{R}, \, \re z \in E_\alpha + \sigma I} \big\Vert \big( ( \sigma P_{\alpha , \sigma} + E_{\alpha} ) - z \big)^{-n} - ( H_{\alpha} - z )^{-n} \big\Vert \leq \frac{\mathrm{C}_{n}}{ \sigma^{n}} ,
\end{equation*}
which is a consequence of the spectral theorem and the choice of $\psi$.
\end{proof}

Likewise, the next proposition follows from Theorem \ref{a21}, Theorem \ref{c17} and Lemma \ref{c18}.

\begin{proposition}\sl \label{c14}
Let $\varphi \in \mathrm{C}_0^\infty( ( 0 , 1 ) ; \mathbb{R} )$. There exists $\alpha_{c} > 0$ such that, for all $s \geq 0$, there exits $\mathrm{C}_{s , \varphi} > 0$ such that, for all $0 \leq \alpha \leq \alpha_c$, $0 < \sigma \leq e_{\mathrm{gap}} / 2$ and $t \in \R$,
\begin{equation*}
\big\Vert \langle B^\sigma \rangle^{-s} e^{-\ii t H_\alpha } \varphi_\sigma( H_\alpha - E_\alpha ) \langle B^\sigma \rangle^{-s} \big\Vert \leq \frac{\mathrm{C}_{s , \varphi}}{\< t \sigma \>^{s}} .
\end{equation*}
\end{proposition}

\begin{proof}
We use the notations of the proof of Proposition \ref{c8} with $I$ such that $\supp ( \varphi ) \subset I$. From the properties of $\psi$, we have
\begin{equation*}
\big\Vert \langle B^{\sigma} \rangle^{- s} e^{- \ii t H_{\alpha}} \varphi_{\sigma} ( H_{\alpha} - E_{\alpha} ) \Phi \big\Vert = \big\Vert \langle B^{\sigma} \rangle^{- s} e^{- \ii t \sigma P_{\alpha , \sigma}} \varphi ( P_{\alpha , \sigma} ) \Phi \big\Vert ,
\end{equation*}
for all $ \Phi \in \mathcal{H}$. On the other hand, combining the uniform Mourre estimate \eqref{a19}, the uniform upper bounds \eqref{a20} with Theorem \ref{a21} and the remark after it, we get
\begin{equation*}
\big\Vert \langle B^{\sigma} \rangle^{- s} e^{- \ii \tau P_{\alpha , \sigma}} \varphi ( P_{\alpha , \sigma} ) \Phi \big\Vert \leq \mathrm{C}_{s , \varphi} \langle \tau \rangle^{- s} \big\Vert \langle B^{\sigma} \rangle^{s} \Phi \big\Vert ,
\end{equation*}
uniformly for $\Phi \in D ( \langle B^{\sigma} \rangle^{s} )$, $0 < \sigma \leq e_{\mathrm{gap}} / 2$, $0 \leq \alpha \leq \alpha_c$ and $\tau \in \R$. The proposition follows from the last two equations.
\end{proof}

\section{Hardy type estimates in Fock space}\label{f4}

\subsection{Hardy estimates}

The classical Hardy estimate in $\mathbb{R}^3$ states that
\begin{align}\label{d13}
\int_{\mathbb{R}^3} \frac{ \vert \varphi(x) \vert^2 }{\vert x\vert^2} \d x \leq 4 \int_{\mathbb{R}^3} \vert \nabla_x \varphi(x) \vert^2 \d x,
\end{align}
for any $\varphi \in \mathrm{C}_0^\infty( \mathbb{R}^3 )$. Since $\mathrm{C}_0^\infty( \mathbb{R}^3 )$ is a core for $\vert\ii \nabla_x\vert = \sqrt{ - \Delta_x }$, this implies that $D(\vert\ii \nabla_x\vert) \subset D(\vert x\vert^{-1})$ and that for all $\varphi \in D(\vert\ii \nabla_x\vert)$, $\Vert \vert x\vert^{-1}\varphi\Vert \leq 2 \Vert \vert \ii \nabla_x \vert \varphi \Vert$. In this subsection we transfer this inequality to Fock space by means of the following lemma. We do not present its proof; (it is essentially the same as the ones of \cite[Lemma A.2]{Ge} and \cite[Proposition 3.4]{GGM}).

\begin{lemma}\sl \label{d14}
Let $n \in \mathbb{N}$. Let $a,b$ be two self-adjoint operators on $\mathrm{L}^2( \mathbb{R}^3 \times \{1,2\} )$ with $b\ge0$, $D(b^n) \subset D(a^n)$ and $\Vert a^n \varphi \Vert \leq \Vert b^n \varphi \Vert$ for all $\varphi \in D(b^n)$. Then $D( \d\Gamma(b)^n ) \subset D( \d\Gamma(a)^n )$ and $\Vert \d\Gamma(a)^n \Phi \Vert \leq \Vert \d\Gamma(b)^n \Phi\Vert$ for all $\Phi \in D( \d\Gamma(b)^n )$.
\end{lemma}

For any vector space $\mathcal{V} \subset \mathrm{L}^2( \mathbb{R}^3 \times \{1,2\} )$, we set
\begin{equation}
\begin{aligned}
\Gamma_{\mathrm{fin}}( \mathcal{V} ) := \big\{ \Phi = ( \Phi^{(0)} & , \Phi^{(1)} , \dots ) \in \Gamma( \mathrm{L}^2( \mathbb{R}^3 \times \{1,2\} )) ;   \\
& \forall n , \ \Phi^{(n)} \in \mathcal{V} \quad \text{and} \quad \exists n_0 , \ \forall n \geq n_0 , \ \Phi^{(n)} = 0 \},
\end{aligned}
\end{equation}
and the number operator in $\mathcal{F}$ is defined by 
\begin{align}
\mathcal{N} := \sum_{\lambda=1,2} \int_{\mathbb{R}^3} a^*_\lambda(k)a_\lambda(k) \d k.
\end{align}

\begin{theorem}[Hardy estimate in Fock space]\sl \label{d17}
For all $\rho>0$, the operator $( \d\Gamma( \vert\ii \nabla_k\vert ) + \rho )^{-1} \mathcal{N}^{\,2}$ defined on $\Gamma_{\mathrm{fin}}( \mathrm{C}_0^\infty( \mathbb{R}^3 ) \times \{1,2\} )$ extends by continuity to a bounded operator on $D( \d\Gamma(\vert k\vert))$ and we have that
\begin{equation*}
\big\Vert ( \d\Gamma( \vert\ii \nabla_k\vert ) + \rho )^{-1} \mathcal{N}^{\,2} \Phi \big\Vert \leq 2 \big\Vert \d\Gamma( \vert k \vert ) \Phi \big\Vert,
\end{equation*}
for all $\Phi \in D( \d\Gamma( \vert k \vert ) )$.
\end{theorem}

\begin{proof}
Let $\rho > 0$ and $\Phi \in \Gamma_{\mathrm{fin}}( \mathrm{C}_0^\infty( \mathbb{R}^3 ) \times \{1,2\} )$. Hardy's inequality \eqref{d13} together with Lemma \ref{d14} with $n=1$ implies that $D( \d\Gamma( \vert \ii \nabla_k \vert ) ) \subset D( \d \Gamma( \vert k\vert^{-1} ) )$ and that
\begin{equation*}
\big\Vert \d \Gamma( \vert k\vert^{-1} ) ( \d \Gamma( \vert \ii \nabla_k \vert ) + \rho )^{-1} \big\Vert \leq 2.
\end{equation*}
Therefore
\begin{align}\label{d19}
\big\Vert ( \d\Gamma( \vert\ii \nabla_k\vert ) + \rho )^{-1} \mathcal{N}^{\,2} \Phi \big\Vert \leq ( 2 + \widetilde{\rho} \rho^{-1} ) \big\Vert ( \d\Gamma( \vert k\vert^{-1} ) + \widetilde{\rho} )^{-1} \mathcal{N}^{\,2} \Phi \big\Vert,
\end{align}
for all $\widetilde{\rho} > 0$. An easy application of the Cauchy-Schwarz inequality shows that
\begin{equation*}
\big\Vert \mathcal{N}^{\,2} \Psi \big\Vert \leq \big\Vert \d\Gamma( \vert k\vert^{-1} ) \d\Gamma( \vert k\vert ) \Psi \big\Vert,
\end{equation*}
for all $\Psi \in \Gamma_{\mathrm{fin}}( \mathrm{C}_0^\infty( \mathbb{R}^3 ) \times \{1,2\} )$. This implies that 
\begin{equation}\label{d20}
\big\Vert ( \d\Gamma( \vert k\vert^{-1} ) + \widetilde{\rho} )^{-1} \mathcal{N}^{\,2} \Phi \big\Vert \leq \big\Vert \d\Gamma(\vert k\vert) \Phi \big\Vert,
\end{equation}
for all $\widetilde{\rho}>0$. Combining \eqref{d19} and \eqref{d20}, we obtain that
\begin{equation*}
\big\Vert ( \d\Gamma( \vert\ii \nabla_k\vert ) + \rho )^{-1} \mathcal{N}^{\,2} \Phi \big\Vert \leq ( 2 + \widetilde{\rho} \rho^{-1} ) \big\Vert \d\Gamma(\vert k\vert) \Phi \big\Vert,
\end{equation*}
for all $\widetilde{\rho} > 0$. Letting $\widetilde{\rho} \to 0$ and using the fact that $\Gamma_{\mathrm{fin}}( \mathrm{C}_0^\infty( \mathbb{R}^3 ) \times \{1,2\} )$ is a core for $\d\Gamma(\vert k\vert)$ conclude the proof of the lemma.
\end{proof}

\begin{remark}\sl
In particular, Theorem \ref{d17} implies the following resolvent estimate away from the real axis: for all $K \Subset \C \setminus \R$, there exists $\mathrm{C}_{K} > 0$ such that
\begin{equation*}
\big\Vert \< X \>^{- \frac{1}{2}} ( \d \Gamma ( \vert k \vert ) - z)^{-1} \< X \>^{- \frac{1}{2}} \Phi \big\Vert \leq \mathrm{C}_{K} \big\Vert \< {\mathcal N} \>^{-2} \Phi \big\Vert ,
\end{equation*}
for all $\Phi \in {\mathcal F}$ and $z \in K$, where $X = \d \Gamma( \vert \ii \nabla_k \vert )$ from \eqref{d21}.
\end{remark}

\begin{lemma}\sl \label{d18}
There exists $\mathrm{C}>0$ such that, for all $\sigma > 0$, $\rho > 0$ and $\Phi \in \mathcal{F}$,
\begin{equation} \label{d22}
\big\Vert ( X + \rho )^{-1} ( \mathds{1}_{ \mathcal{F}_{\ge\sigma} } \otimes \bar \Pi_\Omega ) \Phi \big\Vert \leq \mathrm{C} \sigma \big\Vert ( \mathds{1}_{ \mathcal{F}_{\ge\sigma} } \otimes \bar \Pi_\Omega ) \Phi \big\Vert .
\end{equation}
Moreover, there exists $\mathrm{C}^{\prime} > 0$ such that, for all $\sigma>0$, $\tau>0$, $\rho > 0$ and $\Phi \in \mathcal{F}$,
\begin{equation} \label{d23}
\big\Vert ( X + \rho )^{-1} ( \mathds{1}_{ \mathcal{F}_{\ge\sigma} } \otimes \mathds{1}_{(0,\tau]}( H_f ) ) \Phi \big\Vert \leq \mathrm{C}^{\prime} \tau \big\Vert ( \mathds{1}_{ \mathcal{F}_{\ge\sigma} } \otimes \mathds{1}_{(0,\tau]}( H_f ) ) \Phi \big\Vert .
\end{equation}
\end{lemma}

\begin{proof}
Let $\{e_i\}_{i=1}^\infty$ denote an orthonormal basis of $\mathcal{F}_{\ge\sigma}$. Any state $\Phi$ in $\mathrm{Ran} ( \mathds{1}_{\mathcal{F}_{\ge\sigma}} \otimes \bar \Pi_{\Omega} )$ can be written as $\Phi = \sum_{i=1}^\infty e_i \otimes \Phi_i$ where for all $i \in \N$, $\Phi_i \in \mathrm{Ran} \, \bar \Pi_\Omega$. For $n \in \mathbb{N}$, we have
\begin{equation}
\Phi^{(n)}(k_1,\dots,k_n) = \frac{1}{n!} \sum_{i=1}^\infty \sum_{j=0}^{n-1} \sum_{\tau\in\Sigma_n} e_i^{(j)}(k_{\tau(1)},\dots,k_{\tau(j)})\Phi_i^{(n-j)}(k_{\tau(j+1)},\dots,k_{\tau(n)}),
\end{equation}
and
\begin{align}
\big\Vert \Phi^{(n)} \big\Vert^2 = \sum_{j=0}^{n-1} \sum_{i=1}^\infty \big\Vert \Phi_i^{(n-j)} \big\Vert^2.
\end{align}
Thus, using the fact that, for all $j \in \{1,\dots,n-1\}$,
\begin{equation*}
\big\Vert \big( \vert \ii \nabla_{k_1} \vert + \dots + \vert \ii \nabla_{k_n} \vert + \rho \big)^{-1} \big( \vert \ii \nabla_{k_{\tau(j+1)}} \vert + \dots + \vert \ii \nabla_{k_{\tau(n)}} \vert + \rho \big) \big\Vert \leq 1,
\end{equation*}
we can write
\begin{align}
\big\Vert \big( ( & X + \rho )^{-1} \Phi \big)^{(n)}( k_1,\dots,k_n ) \big\Vert \notag \\
&\leq \frac{1}{n!} \sum_{j=0}^{n-1} \sum_{\tau\in\Sigma_n} \Big \Vert \sum_{i=1}^\infty e_i^{(j)}(k_{\tau(1)},\dots,k_{\tau(j)})  \big( ( X + \rho )^{-1} \Phi_i \big)^{(n-j)}(k_{\tau(j+1)},\dots,k_{\tau(n)}) \Big \Vert \notag \\
&= \sum_{j=0}^{n-1} \Big ( \sum_{i=1}^\infty \big\Vert e_i^{(j)}(k_{1},\dots,k_{j})  \big( ( X + \rho )^{-1} \Phi_i \big)^{(n-j)}(k_{j+1},\dots,k_{n}) \big\Vert^2 \Big )^{\frac{1}{2}} \notag \\
&=\sum_{j=0}^{n-1} \Big ( \sum_{i=1}^\infty \big\Vert \big( ( X + \rho )^{-1} \Phi_i \big)^{(n-j)}(k_{j+1},\dots,k_{n}) \big\Vert^2 \Big )^{\frac{1}{2}}.
\end{align}
Applying Theorem \ref{d17} and using that $\Vert ( H_f \Phi_i )^{(n-j)} \Vert \leq (n-j) \sigma \Vert \Phi_i^{(n-j)} \Vert$, this yields
\begin{align*}
\big\Vert \big( ( X + \rho )^{-1} \Phi \big)^{(n)} \big\Vert & \leq 2 \sum_{j=0}^{n-1} \frac{1}{(n-j)^2} \Big ( \sum_{i=1}^\infty \big\Vert ( H_f \Phi_i)^{(n-j)} \big\Vert^2 \Big )^{\frac{1}{2}}  \\
& \leq 2\sigma \sum_{j=0}^{n-1} \frac{1}{n-j} \Big ( \sum_{i=1}^\infty \big\Vert \Phi_i^{(n-j)} \big\Vert^2 \Big )^{\frac{1}{2}}.
\end{align*}
Applying then the Cauchy-Schwarz inequality, we obtain that
\begin{align}
\big\Vert \big( ( X + \rho )^{-1} \Phi \big)^{(n)} \big\Vert & \leq 2\sigma \Big ( \sum_{j=0}^{n-1} \frac{1}{(n-j)^2} \Big )^{\frac{1}{2}} \Big ( \sum_{j=0}^{n-1} \sum_{i=1}^\infty \big\Vert \Phi_i^{(n-j)} \big\Vert^2 \Big )^{\frac{1}{2}} \notag \\
&\leq 2\sqrt{\frac{\pi^{2}}{6}} \sigma \big\Vert \Phi^{(n)} \big\Vert,
\end{align}
and hence \eqref{d22} is proven. In order to prove \eqref{d23}, it suffices to proceed in the same way, using instead that, if $\Phi \in \mathrm{Ran}(\mathds{1}_{ \mathcal{F}_{\ge\sigma} } \otimes \mathds{1}_{(0,\tau]}( H_f ) )$, then $\Vert ( H_f \Phi_i )^{(n-j)} \Vert \leq \tau \Vert \Phi_i^{(n-j)} \Vert$.
\end{proof}

\begin{corollary}\sl \label{d15}
There exists $\alpha_c>0$ such that, for all $\varphi \in \mathrm{C}_0^\infty( (0, 1 ) ; \mathbb{R} )$, there exists $\mathrm{C}_{\varphi} > 0$ such that, for all $0 \leq \alpha \leq \alpha_c$ and $0 < \sigma \leq e_{\mathrm{gap}} / 2$, 
\begin{equation*}
\big\Vert \< X \>^{-1} \varphi_\sigma( H_{\alpha,\sigma} - E_{\alpha,\sigma} ) \big\Vert \leq \mathrm{C}_\varphi \sigma.
\end{equation*}
\end{corollary}

\begin{proof}
It suffices to use \eqref{g2} and next to apply Lemma \ref{d18}.
\end{proof}

\begin{corollary}\sl \label{c15}
There exists $\alpha_c>0$ such that, for all $\varphi \in \mathrm{C}_0^\infty( (0, 1 ) ; \mathbb{R} )$, there exists $\mathrm{C}_{\varphi} > 0$ such that, for all $0 \leq \alpha \leq \alpha_c$ and $0 < \sigma \leq e_{\mathrm{gap}} / 2$, 
\begin{equation} \label{d16}
\big\Vert \< X \>^{-1} \varphi_\sigma( H_{\alpha} - E_{\alpha} ) \big\Vert \leq \mathrm{C}_{\varphi} \sigma.
\end{equation}
\end{corollary}

\begin{proof}
It suffices to decompose $\varphi_\sigma( H_{\alpha} - E_{\alpha} ) = \varphi_\sigma ( H_{\alpha,\sigma} - E_{\alpha,\sigma} ) + ( \varphi_\sigma( H_{\alpha} - E_{\alpha} ) - \varphi_\sigma( H_{\alpha,\sigma} - E_{\alpha,\sigma} ) )$. Estimate \eqref{d16} is then a consequence of Corollary \ref{d15} together with Proposition \ref{c10}.
\end{proof}

\subsection{Relative bounds on the conjugate operator}

The next lemma is an easy consequence of Lemma \ref{d14}.

\begin{lemma}\sl \label{d11}
There exists $\mathrm{C}>0$ such that, for all $\sigma > 0$, $D( X ) \subset D( B^\sigma )$ and
\begin{equation*}
\big\Vert B^\sigma \< X \>^{-1} \Phi \big\Vert \leq \mathrm{C} \sigma \Vert \Phi \Vert,
\end{equation*}
for all  $\Phi \in \mathcal{F}$.
\end{lemma}

\begin{proof}
Observe that
\begin{equation} \label{d12}
b^\sigma = \frac{ 3 \ii }{ 2 } \big( \eta_\sigma(k) \big)^2 + \ii \eta_\sigma(k) k \cdot (\nabla_k \eta_\sigma) (k) + \ii \big( \eta_\sigma(k) \big)^2 k \cdot \nabla_k .
\end{equation}
Hardy's inequality \eqref{d13} implies that, for all $\varphi \in \mathrm{C}_0^\infty( \mathbb{R}^3 )$,
\begin{equation*}
\big\Vert \big( \eta_\sigma(k) \big)^2 \varphi \big\Vert \leq 2 \sigma \big\Vert \vert \ii \nabla_k \vert \varphi \big\Vert ,
\end{equation*}
and likewise
\begin{equation*}
\big\Vert \eta_\sigma(k) k \cdot (\nabla_k \eta_\sigma) (k) \varphi \big\Vert \leq \mathrm{C} \sigma \big\Vert \vert \ii \nabla_k \vert \varphi \big\Vert .
\end{equation*}
Moreover, we obviously have that
\begin{equation*}
\big\Vert \big( \eta_\sigma(k) \big)^2 k \cdot \nabla_k \varphi \big\Vert \leq \sigma \big\Vert \vert \ii \nabla_k \vert \varphi \big\Vert .
\end{equation*}
Since $\mathrm{C}_0^\infty( \mathbb{R}^3 )$ is a core for $\vert \ii \nabla_k \vert$, \eqref{d12} and the previous estimates imply that $D( \vert \ii \nabla_k \vert ) \subset D( b^\sigma )$ and that, for all $\varphi \in D( \vert \ii \nabla_k \vert )$,
\begin{align} \label{a1}
\Vert b^\sigma \varphi \Vert \leq \mathrm{C} \sigma \big\Vert \vert \ii \nabla_k \vert \varphi \big\Vert .
\end{align}
Applying Lemma \ref{d14} with $a = b^\sigma$ and $b = \mathrm{C} \sigma \vert \ii \nabla_k \vert$ concludes the proof of the lemma.
\end{proof}

\begin{corollary}\sl \label{d1}
There exists $\alpha_c>0$ such that, for all $\varphi \in \mathrm{C}_0^\infty( (0, 1 ) ; \mathbb{R} )$, there exists $\mathrm{C}_{\varphi} > 0$ such that, for all $0 \leq \alpha \leq \alpha_c$ and $0 < \sigma \leq e_{\mathrm{gap}} / 2$,
\begin{equation*}
\big\Vert \< X \>^{-1} \varphi_\sigma( H_{\alpha,\sigma} - E_{\alpha,\sigma} ) B^\sigma \big\Vert \leq \mathrm{C}_{\varphi} \sigma.
\end{equation*}
\end{corollary}

\begin{proof}
Using Lemma \ref{e36}, we can write
\begin{equation*}
\ii \< X \>^{-1} \varphi_\sigma( H_{\alpha,\sigma} - E_{\alpha,\sigma} ) B^\sigma = \< X \>^{-1} \ad_{\ii B^\sigma}^1 ( \varphi_\sigma( H_{\alpha,\sigma} - E_{\alpha,\sigma} ) ) +\ii  \< X \>^{-1} B^\sigma \varphi_\sigma( H_{\alpha,\sigma} - E_{\alpha,\sigma} ),
\end{equation*}
as an identity on $D(B^\sigma)$. Note that, by \eqref{g2} and \eqref{g6}, we have 
\begin{equation*}
\ad_{\ii B^\sigma}^1( \varphi_\sigma( H_{\alpha,\sigma} - E_{\alpha,\sigma} ) ) = \Pi_{\alpha,\ge\sigma} \otimes \big( \d \Gamma( ( \eta_\sigma(k) )^2 |k| ) ( \varphi_\sigma )^{\prime} ( H_f ) \big) ,
\end{equation*}
and hence, since $( \varphi_\sigma )^{\prime} ( H_f ) = \bar \Pi_\Omega ( \varphi_\sigma )^{\prime} ( H_f )$, we can write
\begin{equation*}
\ad_{\ii B^\sigma}^1( \varphi_\sigma( H_{\alpha,\sigma} - E_{\alpha,\sigma} ) ) = ( \mathds{1}_{ \mathcal{F}_{\ge\sigma} } \otimes \bar \Pi_\Omega ) \ad_{\ii B^\sigma}^1( \varphi_\sigma( H_{\alpha,\sigma} - E_{\alpha,\sigma} ) ) .
\end{equation*}
Then, Lemma \ref{d18} and Lemma \ref{e36} imply
\begin{equation*}
\big\Vert \< X \>^{-1} \ad_{\ii B^\sigma}^1( \varphi_\sigma( H_{\alpha,\sigma} - E_{\alpha,\sigma} ) ) \big\Vert \leq \mathrm{C}_\varphi \sigma.
\end{equation*}
Moreover according to Lemma \ref{d11},
\begin{equation*}
\big\Vert \< X \>^{-1} B^\sigma \varphi_\sigma( H_{\alpha,\sigma} - E_{\alpha,\sigma} ) \big\Vert \leq \mathrm{C} \sigma,
\end{equation*}
and hence the proof is complete.
\end{proof}

\begin{corollary}\sl \label{c16}
There exists $\alpha_c>0$ such that, for all $\varphi \in \mathrm{C}_0^\infty( (0, 1 ) ; \mathbb{R} )$ and $\delta > 0$, there exists $\mathrm{C}_{\varphi , \delta} >0$ such that, for all $0 \leq \alpha \leq \alpha_c$ and $0 < \sigma \leq e_{\mathrm{gap}} / 2$,
\begin{equation*}
\big\Vert \< X \>^{-1} \varphi_\sigma( H_{\alpha} - E_{\alpha} ) B^\sigma \big\Vert \leq \mathrm{C}_{\varphi , \delta} \sigma^{1 - \delta} .
\end{equation*}
\end{corollary}

\begin{proof}
As in the proof of Corollary \ref{c15}, we decompose
\begin{align*}
\< X \>^{-1} \varphi_\sigma( H_{\alpha} - E_{\alpha} ) B^\sigma = {}& \< X \>^{-1} \varphi_\sigma( H_{\alpha,\sigma} - E_{\alpha,\sigma} ) B^\sigma \\
&+ \< X \>^{-1} \big( \varphi_\sigma( H_{\alpha} - E_{\alpha} ) - \varphi_\sigma( H_{\alpha,\sigma} - E_{\alpha,\sigma} ) \big) B^\sigma.
\end{align*}
The first term in the right hand side is estimated thanks to Corollary \ref{d1}. As for the second term, it suffices to write (as an identity on $D(B^\sigma)$)
\begin{align*}
\< X \>^{-1} \big( \varphi_\sigma( H_{\alpha} - E_{\alpha} ) & - \varphi_\sigma( H_{\alpha,\sigma} - E_{\alpha,\sigma} ) \big) B^\sigma  \\
={}& \< X \>^{-1} B^\sigma \big( \varphi_\sigma( H_{\alpha} - E_{\alpha} ) - \varphi_\sigma( H_{\alpha,\sigma} - E_{\alpha,\sigma} ) \big)  \\
&+ \< X \>^{-1} \big[ \big( \varphi_\sigma( H_{\alpha} - E_{\alpha} ) - \varphi_\sigma( H_{\alpha,\sigma} - E_{\alpha,\sigma} ) \big) , B^\sigma \big],
\end{align*}
and next to use Lemma \ref{d11}, Proposition \ref{c10} and Proposition \ref{e37}.
\end{proof}

\subsection{Relative bounds on powers of the conjugate operator}

We now estimate $(B^\sigma)^2$ relatively to $\< X \>^2$. Since Hardy's inequality $\| |k|^{-s} \varphi \| \le \mathrm{C}_s \| | \ii \nabla_k |^s \varphi \|$ does not hold for $s \ge 3/2$, we cannot directly apply Lemma \ref{d14}. Nevertheless, some aspects of the following proof are taken from the proof of Lemma \ref{d14}.

\begin{lemma}\sl \label{d2}
There exists $\mathrm{C}>0$ such that, for all $0 < \sigma \le e_{\mathrm{gap}}/2$, $D( X^2 ) \subset D( ( B^\sigma )^2 )$ and
\begin{equation*}
\big\Vert ( B^\sigma )^2 \< X \>^{-2} \Phi \big\Vert \leq \mathrm{C} \sigma \Vert \Phi \Vert,
\end{equation*}
for all $\Phi \in \mathcal{F}$.
\end{lemma}

\begin{proof}
First, we will collect some estimates on powers of $b^{\sigma}$. Since $b^{\sigma}$ is self-adjoint, an interpolation argument and \eqref{a1} imply that there exists $\mathrm{C}>0$ such that
\begin{gather}
\big\Vert \vert \ii \nabla_k \vert^{-\frac{1}{2}} b^\sigma \vert \ii \nabla_k \vert^{-\frac{1}{2}} \big\Vert \leq \mathrm{C} \sigma ,  \label{a2}  \\
\big\Vert \vert \ii \nabla_k \vert^{-1} ( b^\sigma )^{2} \vert \ii \nabla_k \vert^{-1} \big\Vert \leq \mathrm{C} \sigma^{2} ,  \label{a3}
\end{gather}
uniformly for all $\sigma >0$. Here and in the rest of the proof, $\vert \ii \nabla_k \vert^{-1}$ is an abuse of notation. To be rigorous, one has to replace $\vert \ii \nabla_k \vert^{-1}$ by $( \vert \ii \nabla_k \vert + \rho )^{-1}$ with $\rho > 0$, say that the estimates are uniform with respect to $\rho$ and, at the end, let $\rho$ goes to $0$. Since this presents no problem, and to avoid heaviness in the notations, we will omit these technical details and simply write $\vert \ii \nabla_k \vert^{-1}$.

A direct computation shows that $( b^{\sigma} )^{3}$ is a linear combination of terms of the form
\begin{gather*}
r^{(3)} : = \sigma^{3} D_{i} \Big( D_{j} f \Big( \frac{k}{\sigma} \Big) + f \Big( \frac{k}{\sigma} \Big) D_{j} \Big) D_{\ell} ,  \qquad r^{(2)} : = \sigma^{2} D_{i} f \Big( \frac{k}{\sigma} \Big) D_{j} ,   \\
r^{(1)} : = \sigma \Big( D_{i} f \Big( \frac{k}{\sigma} \Big) + f \Big( \frac{k}{\sigma} \Big) D_{i} \Big) ,  \qquad r^{(0)} : = f \Big( \frac{k}{\sigma} \Big) .
\end{gather*}
where $D_{\#}$ is a short-cut for $\frac{1}{\ii} \frac{\partial \ \ }{\partial_{k_{\#}}}$ and $f$ is a real valued $\mathrm{C}_0^\infty( \mathbb{R}^3 )$ function which may change from line to line. Using Hardy's inequality \eqref{d13}, we get
\begin{align}
\big\Vert \vert \ii \nabla_k \vert^{-\frac{1}{2}} r^{(0)} \vert \ii \nabla_k \vert^{-\frac{1}{2}} \big\Vert \leq \mathrm{C} \Big\Vert f \Big( \frac{k}{\sigma} \Big) \vert \ii \nabla_k \vert^{-\frac{1}{2}} \Big\Vert^{2} \leq \mathrm{C} \sigma .  \label{a4}
\end{align}
Since the operator $r^{(1)}$ is similar to $b^{\sigma}$, we can proceed as in \eqref{a1} to obtain
\begin{align*}
\Vert r^{(1)} \varphi \Vert \leq \mathrm{C} \sigma \big\Vert \vert \ii \nabla_k \vert  \varphi \big\Vert ,
\end{align*}
for all $\varphi \in D( \vert \ii \nabla_k \vert )$. Using that $r^{(1)}$ is self-adjoint, an interpolation argument gives
\begin{align}
\big\Vert \vert \ii \nabla_k \vert^{-\frac{1}{2}} r^{(1)} \vert \ii \nabla_k \vert^{-\frac{1}{2}} \big\Vert \leq \mathrm{C} \sigma .  \label{a5}
\end{align}
On the other hand, since $D_{\#}$ and $\vert \ii \nabla_k \vert$ commute, \eqref{a4} implies
\begin{align}
\big\Vert \vert \ii \nabla_k \vert^{-\frac{3}{2}} r^{(2)} \vert \ii \nabla_k \vert^{-\frac{3}{2}} \big\Vert &\leq \sigma^{2} \big\Vert \vert \ii \nabla_k \vert^{-\frac{1}{2}} r^{(0)} \vert \ii \nabla_k \vert^{-\frac{1}{2}} \big\Vert \big\Vert D_{i} \vert \ii \nabla_k \vert^{-1} \big\Vert \big\Vert D_{j} \vert \ii \nabla_k \vert^{-1} \big\Vert  \nonumber  \\
&\leq \mathrm{C} \sigma^{3} .  \label{a6}
\end{align}
Eventually, using \eqref{a5}, we obtain
\begin{align}
\big\Vert \vert \ii \nabla_k \vert^{-\frac{3}{2}} r^{(3)} \vert \ii \nabla_k \vert^{-\frac{3}{2}} \big\Vert &\leq \sigma^{2} \big\Vert \vert \ii \nabla_k \vert^{-\frac{1}{2}} r^{(1)} \vert \ii \nabla_k \vert^{-\frac{1}{2}} \big\Vert \big\Vert D_{i} \vert \ii \nabla_k \vert^{-1} \big\Vert \big\Vert D_{\ell} \vert \ii \nabla_k \vert^{-1} \big\Vert    \nonumber \\
&\leq \mathrm{C} \sigma^{3} .  \label{a7}
\end{align}

The same way, $(b^{\sigma})^{4}$ is a linear combination of terms of the form $r^{(0)}$, $r^{(1)}$, $r^{(2)}$, $r^{(3)}$ and
\begin{align*}
r^{(4)} : = \sigma^{4} D_{i} D_{j} f \Big( \frac{k}{\sigma} \Big) D_{\ell} D_{m} .
\end{align*}
The terms $r^{(0)} , \ldots , r^{(3)}$ have already been studied in \eqref{a4}--\eqref{a7}, but we will need other estimates to control $(b^{\sigma})^{4}$. Using Hardy's inequality \eqref{d13}, we have
\begin{align}
\big\Vert \vert \ii \nabla_k \vert^{-1} r^{(0)} \vert \ii \nabla_k \vert^{-1} \big\Vert & \leq \Big\Vert f \Big( \frac{k}{\sigma} \Big) \vert \ii \nabla_k \vert^{-1} \Big\Vert^{2} \leq \mathrm{C} \sigma^{2} ,  \label{a8}  \\
\big\Vert \vert \ii \nabla_k \vert^{-1} r^{(1)} \vert \ii \nabla_k \vert^{-1} \big\Vert & \leq 2 \sigma \Big\Vert f \Big( \frac{k}{\sigma} \Big) \vert \ii \nabla_k \vert^{-1} \Big\Vert \big\Vert D_{i} \vert \ii \nabla_k \vert^{-1} \big\Vert \leq \mathrm{C} \sigma^{2} , \label{a9}  \\
\big\Vert \vert \ii \nabla_k \vert^{-1} r^{(2)} \vert \ii \nabla_k \vert^{-1} \big\Vert & \leq \mathrm{C} \sigma^{2} \big\Vert D_{i} \vert \ii \nabla_k \vert^{-1} \big\Vert \big\Vert D_{j} \vert \ii \nabla_k \vert^{-1} \big\Vert \leq \mathrm{C} \sigma^{2} , \label{a10}
\end{align}
and
\begin{align}
\big\Vert \vert \ii \nabla_k \vert^{-2} r^{(3)} \vert \ii \nabla_k \vert^{-2} \big\Vert &\leq 2 \sigma^{3} \big\Vert D_{i} \vert \ii \nabla_k \vert^{-1} \big\Vert \big\Vert D_{j} \vert \ii \nabla_k \vert^{-1} \big\Vert \Big\Vert f \Big( \frac{k}{\sigma} \Big) \vert \ii \nabla_k \vert^{-1} \Big\Vert \big\Vert D_{\ell} \vert \ii \nabla_k \vert^{-1} \big\Vert   \nonumber \\
&\leq \mathrm{C} \sigma^{4} ,  \label{a11}    \\
\big\Vert \vert \ii \nabla_k \vert^{-2} r^{(4)} \vert \ii \nabla_k \vert^{-2} \big\Vert &\leq \mathrm{C} \sigma^{4} \big\Vert D_{i} \vert \ii \nabla_k \vert^{-1} \big\Vert \big\Vert D_{j} \vert \ii \nabla_k \vert^{-1} \big\Vert \big\Vert D_{\ell} \vert \ii \nabla_k \vert^{-1} \big\Vert \big\Vert D_{m} \vert \ii \nabla_k \vert^{-1} \big\Vert   \nonumber \\
&\leq \mathrm{C} \sigma^{4} .    \label{a17}
\end{align}

We now estimate $( B^{\sigma} )^{2} \Phi$ for $\Phi \in \Gamma_{\mathrm{fin}}( \mathrm{C}_0^\infty( \mathbb{R}^3 \times \{1,2\} ) )$. Assume that $n \in \mathbb{N}$. For $i \in \{ 1 , \ldots , n \}$, let $b_{i}^{\sigma}$ denote the operator $b^{\sigma}$ acting on the variable $k_{i}$. We can write
\begin{align}
\big\Vert \big( (B^{\sigma} )^{2} \Phi \big)^{(n)} ( k_{1} , \dots , k_{n} ) \big\Vert^{2}  &= \sum_{i_{1} , i_{2} , i_{3} , i_{4} \in \{ 1 , \ldots , n \}} \big\langle b_{i_{1}}^{\sigma} b_{i_{2}}^{\sigma} b_{i_{3}}^{\sigma} b_{i_{4}}^{\sigma} \Phi^{(n)} , \Phi^{(n)} \big\rangle    \nonumber\\
&= {\mathcal A}_{1,2} + {\mathcal A}_{3} + {\mathcal A}_{4} ,   \label{a12}
\end{align}
where
\begin{align*}
{\mathcal A}_{\#} = \sum_{i_{1} , i_{2} , i_{3} , i_{4} \in A_{\#}} \big\langle b_{i_{1}}^{\sigma} b_{i_{2}}^{\sigma} b_{i_{3}}^{\sigma} b_{i_{4}}^{\sigma} \Phi^{(n)} , \Phi^{(n)} \big\rangle ,
\end{align*}
and
\begin{align*}
A_{4} &= \big\{ ( i_{1} , i_{2} , i_{3} , i_{4} ) ; \ \text{all the } i_{\#}\text{'s} \text{ are equal} \big\} ,    \\
A_{3} &= \big\{ ( i_{1} , i_{2} , i_{3} , i_{4} ) ; \ \text{three of the } i_{\#}\text{'s} \text{ are equal} \big\} \setminus A_{4} ,  \\
A_{1,2} &= \{ ( i_{1} , i_{2} , i_{3} , i_{4} ) \} \setminus ( A_{3} \cup A_{4} ) .
\end{align*}

Note that for $( i_{1} , i_{2} , i_{3} , i_{4} ) \in A_{1,2}$, only $b^{\sigma}_{\#}$ and $( b^{\sigma}_{\#} )^{2}$ can appear in $ b_{i_{1}}^{\sigma} b_{i_{2}}^{\sigma} b_{i_{3}}^{\sigma} b_{i_{4}}^{\sigma}$ and not $( b^{\sigma}_{\#} )^{3}$ or $( b^{\sigma}_{\#} )^{4}$. Thus, using that the operators $b_{i}^{\sigma}$ and $\vert \ii \nabla_{k_{i}} \vert$ commute with the operators $b_{j}^{\sigma}$ and $\vert \ii \nabla_{k_{j}} \vert$ for $i \neq j$, the estimates \eqref{a2}--\eqref{a3} imply
\begin{align}
\big\vert {\mathcal A}_{1,2} \big\vert \leq{}& \sum_{i_{1} , i_{2} , i_{3} , i_{4} \in A_{1,2}} \Big\vert \Big\langle \vert \ii \nabla_{k_{i_{1}}} \vert^{-\frac{1}{2}} \cdots \vert \ii \nabla_{k_{i_{4}}} \vert^{-\frac{1}{2}} b_{i_{1}}^{\sigma} b_{i_{2}}^{\sigma} b_{i_{3}}^{\sigma} b_{i_{4}}^{\sigma} \vert \ii \nabla_{k_{i_{1}}} \vert^{-\frac{1}{2}} \cdots \vert \ii \nabla_{k_{i_{4}}} \vert^{-\frac{1}{2}}     \nonumber \\
&\qquad \qquad \qquad \qquad \qquad \vert \ii \nabla_{k_{i_{1}}} \vert^{\frac{1}{2}} \cdots \vert \ii \nabla_{k_{i_{4}}} \vert^{\frac{1}{2}} \Phi^{(n)} , \vert \ii \nabla_{k_{i_{1}}} \vert^{\frac{1}{2}} \cdots \vert \ii \nabla_{k_{i_{4}}} \vert^{\frac{1}{2}} \Phi^{(n)} \Big\rangle \Big\vert       \nonumber \\
\leq{}& \mathrm{C} \sigma^{4} \sum_{i_{1} , i_{2} , i_{3} , i_{4} \in A_{1,2}} \big\Vert \vert \ii \nabla_{k_{i_{1}}} \vert^{\frac{1}{2}} \vert \ii \nabla_{k_{i_{2}}} \vert^{\frac{1}{2}} \vert \ii \nabla_{k_{i_{3}}} \vert^{\frac{1}{2}} \vert \ii \nabla_{k_{i_{4}}} \vert^{\frac{1}{2}} \Phi^{(n)} \big\Vert^{2}      \nonumber \\
\leq{}& \mathrm{C} \sigma^{4} \sum_{i_{1} , i_{2} , i_{3} , i_{4}} \big\langle \vert \ii \nabla_{k_{i_{1}}} \vert \vert \ii \nabla_{k_{i_{2}}} \vert \vert \ii \nabla_{k_{i_{3}}} \vert \vert \ii \nabla_{k_{i_{4}}} \vert \Phi^{(n)} , \Phi^{(n)} \big\rangle    \nonumber \\
={}& \mathrm{C} \sigma^{4} \big\Vert \big( \d \Gamma ( \vert \ii \nabla_{k} \vert )^{2} \Phi \big)^{(n)} \big\Vert^{2} . \label{a14}
\end{align}

On the other hand, using that $( b^{\sigma} )^{3}$ is a linear combination of terms of the form $r^{(0)}$, $r^{(1)}$, $r^{(2)}$ and $r^{(3)}$, we can write
\begin{align*}
\big\vert {\mathcal A}_{3} \big\vert &= 4 \Big\vert \sum_{i \neq j} \big\langle ( b_{i}^{\sigma} )^{3} b_{j}^{\sigma} \Phi^{(n)} , \Phi^{(n)} \big\rangle \Big\vert  \\
&\leq \mathrm{C} \sum_{i \neq j} \sum_{\text{finite}} \big\vert \big\langle ( r_{i}^{(0)} + r_{i}^{(1)} ) b_{j}^{\sigma} \Phi^{(n)} , \Phi^{(n)} \big\rangle \big\vert + \mathrm{C} \sum_{i \neq j} \sum_{\text{finite}} \big\vert \big\langle ( r_{i}^{(2)} + r_{i}^{(3)} ) b_{j}^{\sigma} \Phi^{(n)} , \Phi^{(n)} \big\rangle \big\vert ,
\end{align*}
where $r_{\#}^{( \ast )}$ denoted the operator $r^{( \ast )}$ in the $k_{\#}$-variables. As in \eqref{a14}, \eqref{a4}--\eqref{a5} yield
\begin{align*}
\big\vert \big\langle ( r_{i}^{(0)} + r_{i}^{(1)} ) b_{j}^{\sigma} \Phi^{(n)} , \Phi^{(n)} \big\rangle \big\vert &\leq \big\vert \big\langle \vert \ii \nabla_{k_{i}} \vert^{-\frac{1}{2}} ( r_{i}^{(0)} + r_{i}^{(1)} ) \vert \ii \nabla_{k_{i}} \vert^{-\frac{1}{2}} \vert \ii \nabla_{k_{j}} \vert^{-\frac{1}{2}} b_{j}^{\sigma} \vert \ii \nabla_{k_{j}} \vert^{-\frac{1}{2}}   \\
&\qquad \qquad \qquad \quad \vert \ii \nabla_{k_{i}} \vert^{\frac{1}{2}} \vert \ii \nabla_{k_{j}} \vert^{\frac{1}{2}} \Phi^{(n)} , \vert \ii \nabla_{k_{i}} \vert^{\frac{1}{2}} \vert \ii \nabla_{k_{j}} \vert^{\frac{1}{2}} \Phi^{(n)} \big\rangle \big\vert  \\
&\leq \mathrm{C} \sigma^{2} \big\Vert \vert \ii \nabla_{k_{i}} \vert^{\frac{1}{2}} \vert \ii \nabla_{k_{j}} \vert^{\frac{1}{2}} \Phi^{(n)} \big\Vert^{2}  \\
&= \mathrm{C} \sigma^{2} \big\langle \vert \ii \nabla_{k_{i}} \vert \vert \ii \nabla_{k_{j}} \vert \Phi^{(n)} , \Phi^{(n)} \big\rangle .
\end{align*}
The same way, \eqref{a6}--\eqref{a7} give
\begin{align*}
\big\vert \big\langle ( r_{i}^{(2)} + r_{i}^{(3)} ) b_{j}^{\sigma} \Phi^{(n)} , \Phi^{(n)} \big\rangle \big\vert &\leq \big\vert \big\langle \vert \ii \nabla_{k_{i}} \vert^{-\frac{3}{2}} ( r_{i}^{(2)} + r_{i}^{(3)} ) \vert \ii \nabla_{k_{i}} \vert^{-\frac{3}{2}} \vert \ii \nabla_{k_{j}} \vert^{-\frac{1}{2}} b_{j}^{\sigma} \vert \ii \nabla_{k_{j}} \vert^{-\frac{1}{2}}   \\
&\qquad \qquad \qquad \quad \vert \ii \nabla_{k_{i}} \vert^{\frac{3}{2}} \vert \ii \nabla_{k_{j}} \vert^{\frac{1}{2}} \Phi^{(n)} , \vert \ii \nabla_{k_{i}} \vert^{\frac{3}{2}} \vert \ii \nabla_{k_{j}} \vert^{\frac{1}{2}} \Phi^{(n)} \big\rangle \big\vert  \\
&\leq \mathrm{C} \sigma^{4} \big\Vert \vert \ii \nabla_{k_{i}} \vert^{\frac{3}{2}} \vert \ii \nabla_{k_{j}} \vert^{\frac{1}{2}} \Phi^{(n)} \big\Vert^{2}  \\
&= \mathrm{C} \sigma^{4} \big\langle \vert \ii \nabla_{k_{i}} \vert^{3} \vert \ii \nabla_{k_{j}} \vert \Phi^{(n)} , \Phi^{(n)} \big\rangle .
\end{align*}
Combining the previous estimates, we obtain
\begin{align}
\big\vert {\mathcal A}_{3} \big\vert &\leq \mathrm{C} \sigma^{2} \sum_{i \neq j} \big\langle \vert \ii \nabla_{k_{i}} \vert \vert \ii \nabla_{k_{j}} \vert \Phi^{(n)} , \Phi^{(n)} \big\rangle + \mathrm{C} \sigma^{4} \sum_{i \neq j} \big\langle \vert \ii \nabla_{k_{i}} \vert^{3} \vert \ii \nabla_{k_{j}} \vert \Phi^{(n)} , \Phi^{(n)} \big\rangle \nonumber  \\
&\leq \mathrm{C} \sigma^{2} \sum_{i_{1} , i_{2}} \big\langle \vert \ii \nabla_{k_{i_{1}}} \vert \vert \ii \nabla_{k_{i_{2}}} \vert \Phi^{(n)} , \Phi^{(n)} \big\rangle + \mathrm{C} \sigma^{4} \sum_{i_{1} , i_{2} , i_{3} , i_{4}} \big\langle \vert \ii \nabla_{k_{i_{1}}} \vert \vert \ii \nabla_{k_{i_{2}}} \vert \vert \ii \nabla_{k_{i_{3}}} \vert \vert \ii \nabla_{k_{i_{4}}} \vert \Phi^{(n)} , \Phi^{(n)} \big\rangle  \nonumber   \\
&= \mathrm{C} \sigma^{2} \big\Vert \big( \d \Gamma ( \vert \ii \nabla_{k} \vert ) \Phi \big)^{(n)} \big\Vert^{2} + \mathrm{C} \sigma^{4} \big\Vert \big( \d \Gamma ( \vert \ii \nabla_{k} \vert )^{2} \Phi \big)^{(n)} \big\Vert^{2} . \label{a15}
\end{align}

Eventually, we have
\begin{align*}
\big\vert {\mathcal A}_{4} \big\vert &\leq \sum_{i} \big\vert \big\langle ( b_{i}^{\sigma} )^{4} \Phi^{(n)} , \Phi^{(n)} \big\rangle \big\vert   \\
&\leq \mathrm{C} \sum_{i} \sum_{\text{finite}} \big\vert \big\langle ( r_{i}^{(0)} + r_{i}^{(1)} + r_{i}^{(2)} ) \Phi^{(n)} , \Phi^{(n)} \big\rangle \big\vert + \mathrm{C} \sum_{i} \sum_{\text{finite}} \big\vert \big\langle ( r_{i}^{(3)} + r_{i}^{(4)} ) \Phi^{(n)} , \Phi^{(n)} \big\rangle \big\vert .
\end{align*}
Using the estimates \eqref{a8}--\eqref{a17} and proceeding as in the proof of \eqref{a15}, we get
\begin{align}
\big\vert {\mathcal A}_{4} \big\vert &\leq \mathrm{C} \sigma^{2} \sum_{i} \big\langle \vert \ii \nabla_{k_{i}} \vert^{2} \Phi^{(n)} , \Phi^{(n)} \big\rangle + \mathrm{C} \sigma^{4} \sum_{i} \big\langle \vert \ii \nabla_{k_{i}} \vert^{4} \Phi^{(n)} , \Phi^{(n)} \big\rangle \nonumber  \\
&\leq \mathrm{C} \sigma^{2} \sum_{i_{1} , i_{2}} \big\langle \vert \ii \nabla_{k_{i_{1}}} \vert \vert \ii \nabla_{k_{i_{2}}} \vert \Phi^{(n)} , \Phi^{(n)} \big\rangle + \mathrm{C} \sigma^{4} \sum_{i_{1} , i_{2} , i_{3} , i_{4}} \big\langle \vert \ii \nabla_{k_{i_{1}}} \vert \vert \ii \nabla_{k_{i_{2}}} \vert \vert \ii \nabla_{k_{i_{3}}} \vert \vert \ii \nabla_{k_{i_{4}}} \vert \Phi^{(n)} , \Phi^{(n)} \big\rangle  \nonumber   \\
&= \mathrm{C} \sigma^{2} \big\Vert \big( \d \Gamma ( \vert \ii \nabla_{k} \vert ) \Phi \big)^{(n)} \big\Vert^{2} + \mathrm{C} \sigma^{4} \big\Vert \big( \d \Gamma ( \vert \ii \nabla_{k} \vert )^{2} \Phi \big)^{(n)} \big\Vert^{2} .    \label{a16}
\end{align}

Combining \eqref{a12} with the estimates \eqref{a14}, \eqref{a15} and \eqref{a16}, we eventually obtain
\begin{align*}
\big\Vert \big( (B^{\sigma} )^{2} \Phi \big)^{(n)} \big\Vert &\leq \mathrm{C} \sigma^{2} \big\Vert \big( \d \Gamma ( \vert \ii \nabla_{k} \vert ) \Phi \big)^{(n)} \big\Vert^{2} + \mathrm{C} \sigma^{4} \big\Vert \big( \d \Gamma ( \vert \ii \nabla_{k} \vert )^{2} \Phi \big)^{(n)} \big\Vert^{2} \\
&\leq \mathrm{C} \sigma^{2} \big\Vert \big( \big( \d \Gamma( | \ii \nabla_k | ) + 1 \big)^{2} \Phi \big)^{(n)} \big\Vert^{2} ,
\end{align*}
with $\mathrm{C}>0$ uniform with respect to $n\in \mathbb{N}$ and $\sigma >0$. Since this estimate is trivial for $n=0$ and since $\Gamma_{\mathrm{fin}}( \mathrm{C}_0^\infty( \mathbb{R}^3 \times \{1,2\} ) )$ is a core for $D( \d \Gamma( | \ii \nabla_k | )^2 )$, the lemma follows.
\end{proof}

\begin{corollary}\sl \label{d3}
There exists $\alpha_c>0$ such that, for all $\varphi \in \mathrm{C}_0^\infty( (0, 1 ) ; \mathbb{R} )$, there exists $\mathrm{C}_{\varphi} > 0$ such that, for all $0 \leq \alpha \leq \alpha_c$ and $0 < \sigma \leq e_{\mathrm{gap}} / 2$,
\begin{equation*}
\big\Vert \< X \>^{-2} \varphi_\sigma( H_{\alpha,\sigma} - E_{\alpha,\sigma} ) ( B^\sigma )^2 \big\Vert \leq \mathrm{C}_{\varphi} \sigma.
\end{equation*}
\end{corollary}

\begin{proof}
By Lemma \eqref{e36}, we can write as an identity on $D((B^\sigma)^2)$:
\begin{align}
\varphi_\sigma( H_{\alpha,\sigma} - E_{\alpha,\sigma} ) ( \ii B^\sigma )^2 ={}& \ad_{\ii B^\sigma}^2( \varphi_\sigma( H_{\alpha,\sigma} - E_{\alpha,\sigma} ) ) + 2 \ii B^\sigma \ad_{\ii B^\sigma}^1 ( \varphi_\sigma( H_{\alpha,\sigma} - E_{\alpha,\sigma} ) )     \nonumber \\
&+ ( \ii B^\sigma )^2 \varphi_\sigma( H_{\alpha,\sigma} - E_{\alpha,\sigma} ) .   \label{d4}
\end{align}
As in the proof of Corollary \ref{d1}, we have that
\begin{equation*}
\ad_{\ii B^\sigma}^2( \varphi_\sigma( H_{\alpha,\sigma} - E_{\alpha,\sigma} ) ) = ( \mathds{1}_{ \mathcal{H}_{\ge\sigma} } \otimes \bar \Pi_\Omega ) \ad_{\ii B^\sigma}^2( \varphi_\sigma( H_{\alpha,\sigma} - E_{\alpha,\sigma} ) ),
\end{equation*}
and hence it follows from Lemma \ref{d18} and Lemma \ref{e36} that
\begin{align}\label{d5}
\big\Vert \< X \>^{-2} \ad_{\ii B^\sigma}^2( \varphi_\sigma( H_{\alpha,\sigma} - E_{\alpha,\sigma} ) ) \big\Vert \leq \mathrm{C}_\varphi \sigma.
\end{align}
On the other hand, using Lemma \ref{d11} and Lemma \ref{e36}, we obtain that
\begin{align}\label{d6}
\big\Vert \< X \>^{-2} B^\sigma \ad_{\ii B^\sigma}^1( \varphi_\sigma( H_{\alpha,\sigma} - E_{\alpha,\sigma} ) ) \big\Vert \leq \mathrm{C}_\varphi \sigma.
\end{align}
Eventually, Lemma \ref{d2} gives
\begin{equation} \label{g7}
\big\Vert \< X \>^{-2} ( \ii B^\sigma )^2 \varphi_\sigma( H_{\alpha,\sigma} - E_{\alpha,\sigma} ) \big\Vert \leq \mathrm{C}_{\varphi} \sigma .
\end{equation}
The lemma now follows from \eqref{d4}, \eqref{d5}, \eqref{d6} and \eqref{g7}.
\end{proof}

\begin{corollary}\sl \label{c9}
There exists $\alpha_c>0$ such that, for all $\varphi \in \mathrm{C}_0^\infty( (0, 1 ) ; \mathbb{R} )$ and $\delta > 0$, there exists $\mathrm{C}_{\varphi , \delta} > 0$ such that, for all $0 \leq \alpha \leq \alpha_c$ and $0 < \sigma \leq e_{\mathrm{gap}} / 2$,
\begin{equation*}
\big\Vert \< X \>^{-2} \varphi_\sigma( H_{\alpha} - E_{\alpha} ) ( B^\sigma )^2 \big\Vert \leq \mathrm{C}_{\varphi , \delta} \sigma^{1 - \delta} .
\end{equation*}
\end{corollary}

\begin{proof}
The result can be proven exactly in the same way as Corollary \ref{c16}, using Lemma \ref{d2}, Corollary \ref{d3}, Proposition \ref{c10} and Proposition \ref{e37}.
\end{proof}

\section{Proof of the main theorems}\label{f5}

\subsection{Proof of the limiting absorption principle}

We are now ready to prove Theorem \ref{d28}. We can assume that $s \in ( 1/2 , 1 ]$. We define the set of dyadic numbers
\begin{equation*}
{\mathcal D} : = \{ 2^{-n} ; \ n \in \Z \text{ and } 2^{- n} \leq e_{\mathrm{gap}} / 2 \} .
\end{equation*}
Consider $\varphi \in \mathrm{C}_0^\infty( ( 0 , 1 ) ; \R )$ satisfying
\begin{equation} \label{c3}
\forall x \in (0,e_{\mathrm{gap}}/3 ], \qquad \sum_{\sigma \in {\mathcal D}} \varphi_\sigma ( x ) = 1 ,
\end{equation}
and define $\overline{\varphi} \in \mathrm{C}^\infty ( \R ; \R )$ by $\overline{\varphi} (x) =0$ for $x \leq 0$ and
\begin{equation} \label{c4}
\forall x > 0 , \qquad \overline{\varphi} (x) = 1 - \sum_{\sigma \in {\mathcal D}} \varphi_\sigma ( x ) .
\end{equation}
In particular, $\supp ( \overline{\varphi} ) \subset [ e_{\mathrm{gap}}/3 , + \infty )$. Let $\widetilde{\varphi} \in \mathrm{C}_0^\infty( (0, 1 ) ; \mathbb{R} )$ be such that $\widetilde{\varphi} \varphi = \varphi$.

For $\re z \leq E_{\alpha} + e_{\mathrm{gap}} / 4$, the properties of the support of $\overline{\varphi}$ and the spectral theorem give
\begin{align}
{\mathcal J} : ={}& \big\Vert \< X \>^{-s} ( H_\alpha - z )^{-1} \bar{\Pi}_\alpha \< X \>^{-s} \big\Vert \nonumber \\
\leq{}& \sum_{\sigma \in {\mathcal D}} {\mathcal J}_{\sigma} + \big\Vert \< X \>^{-s} ( H_\alpha - z )^{-1} \overline{\varphi} ( H_\alpha - E_\alpha ) \< X \>^{-s} \big\Vert  \nonumber \\
\leq{}& \sum_{\sigma \in {\mathcal D}} {\mathcal J}_{\sigma} + \mathrm{C}   \label{c5}
\end{align}
with $\mathrm{C} > 0$ and
\begin{equation*}
{\mathcal J}_{\sigma} : = \big\Vert \< X \>^{-s} ( H_\alpha - z )^{-1} \varphi_{\sigma} ( H_\alpha - E_\alpha ) \< X \>^{-s} \big\Vert .
\end{equation*}
Theorem \ref{d28} follows from \eqref{c5}, the next lemma and the assumption $s > 1/2$.

\begin{lemma}\sl \label{c7}
There exists $\alpha_{c} > 0$ such that, for all $s \in ( 1/2 , 1 ]$ and $\delta > 0$, there exists $\mathrm{C}_{s , \delta} > 0$ such that, for all $0 \leq \alpha \leq \alpha_c$, $\re z \leq E_{\alpha} + e_{\mathrm{gap}} / 4$, $\im z \neq 0$ and $\sigma \in {\mathcal D}$,
\begin{equation*}
{\mathcal J}_{\sigma} \leq \mathrm{C}_{s , \delta} \sigma^{2 s -1 - \delta} .
\end{equation*}
\end{lemma}

\begin{proof}[Proof of Lemma \ref{c7}]
Let $M > 2$ be a large enough constant such that $\supp ( \varphi ) \subset [ 2 /M , M /2 ]$. We define
\begin{equation*}
\Delta : = {\mathcal D} \cap \big[ \re ( z - E_{\alpha} ) / M , M \re ( z - E_{\alpha} ) \big] .
\end{equation*}
In particular, $\Delta = \emptyset$ when $\re z \leq E_{\alpha}$. We distinguish between different cases.

Assume first that $\sigma \in {\mathcal D} \setminus \Delta$. For $h$ in the support of $\varphi_{\sigma} ( \cdot - E_{\alpha} )$, we have $\vert h -z \vert^{-1} \leq \mathrm{C} \sigma^{-1}$. Then, the spectral theorem gives
\begin{equation} \label{c6}
\big\Vert ( H_\alpha - z )^{-1} \varphi_{\sigma} ( H_\alpha - E_\alpha ) \big\Vert \leq \mathrm{C} \sigma^{-1},
\end{equation}
for $\sigma \in {\mathcal D} \setminus \Delta$. Corollary \ref{c15} (with an interpolation argument) and \eqref{c6} yield
\begin{equation*}
{\mathcal J}_{\sigma} \leq \big\Vert ( H_\alpha - z )^{-1} \varphi_{\sigma} ( H_\alpha - E_\alpha ) \big\Vert \big\Vert \< X \>^{- s} \widetilde{\varphi}_{\sigma} ( H_\alpha - E_\alpha ) \big\Vert^{2} \leq \mathrm{C} \sigma^{-1} \sigma^{2 s} = \mathrm{C} \sigma^{2 s - 1} .
\end{equation*}

We now assume that $\sigma \in \Delta$, that is $\re z \in E_\alpha + \sigma [ 1/M , M ]$. From Proposition \ref{c8} for $n=1$, Corollary \ref{c15} and Corollary \ref{c16} (with an interpolation argument), we get
\begin{align*}
{\mathcal J}_{\sigma} &\leq \big\Vert \< X \>^{- s} \widetilde{\varphi}_{\sigma} ( H_\alpha - E_\alpha ) \< B^\sigma \>^{s} \big\Vert \big\Vert \< B^\sigma \>^{- s} ( H_\alpha - z )^{-1} \< B^\sigma \>^{- s} \big\Vert \big\Vert \< B^\sigma \>^{s} \varphi_{\sigma} ( H_\alpha - E_\alpha ) \< X \>^{- s} \big\Vert \\
&\leq \mathrm{C}_{s , \delta} \sigma^{s - \delta} \sigma^{-1} \sigma^{s - \delta} = \mathrm{C}_{s , \delta} \sigma^{2s-1 - 2 \delta} ,
\end{align*}
which finishes the proof of the lemma.
\end{proof}

\subsection{Proof of the smoothness of the resolvent}

We now show the H\"{o}lder continuity of the resolvent. Let $z , z^{\prime} \in \C \setminus \R$ with $\re z, \re z^{\prime} \leq E_{\alpha} + e_{ \mathrm{gap} } / 4$ and $\im z \cdot \im z^{\prime} > 0$. We can assume that $\re z, \re z^{\prime} \geq - 1$ and $-1 \le \im z , \im z' \le 1$. In the following, $z^{\#}$ will denote either $z$ or $z^{\prime}$. We have to estimate
\begin{equation} \label{b4}
{\mathcal K} : = \big\Vert \< X \>^{- s} \big( ( H_\alpha - z )^{-1} - ( H_\alpha - z^{\prime} )^{-1} \big) \bar{\Pi}_\alpha \< X \>^{- s} \big\Vert .
\end{equation}
Using \eqref{c3}--\eqref{c4} and the spectral theorem, we obtain, as in \eqref{c5},
\begin{align} \label{b1}
{\mathcal K} &\leq \sum_{\sigma \in {\mathcal D}} {\mathcal K}_{\sigma} + \big\Vert \< X \>^{- s} (z - z^{\prime}) ( H_\alpha - z )^{-1}  ( H_\alpha - z^{\prime} )^{-1} \overline{\varphi} ( H_\alpha - E_\alpha ) \< X \>^{- s} \big\Vert  \nonumber  \\
&\leq \sum_{\sigma \in {\mathcal D}} {\mathcal K}_{\sigma} + \mathrm{C} \vert z - z^{\prime} \vert ,
\end{align}
with $\mathrm{C} > 0$ and
\begin{equation*}
{\mathcal K}_{\sigma} = \big\Vert \< X \>^{- s} \big( ( H_\alpha - z )^{-1} - ( H_\alpha - z^{\prime} )^{-1} \big) \varphi_{\sigma} ( H_\alpha - E_\alpha ) \< X \>^{- s} \big\Vert .
\end{equation*}

Let $M > 2$ be a large enough constant such that $\supp ( \varphi ) \subset [ 2 /M , M /2 ]$. As before, we define the set of dyadic numbers
\begin{equation*}
\Delta^{\#} : = {\mathcal D} \cap \big[ \re ( z^{\#} - E_{\alpha} ) / M , M \re ( z^{\#} - E_{\alpha} ) \big] .
\end{equation*}
The ${\mathcal K}_{\sigma}$'s satisfy

\begin{lemma}\sl \label{b7}
There exists $\alpha_{c} > 0$ such that, for all $s \in ( 1/2 , 3/2 )$ and $\varepsilon > 0$, there exists $\mathrm{C}_{s , \varepsilon} > 0$ such that, for all $0 \leq \alpha \leq \alpha_c$ and $z , z^{\prime} \in \C \setminus \R$ with $- 1 \leq \re z, \re z^{\prime} \leq E_{\alpha} + e_{ \mathrm{gap} } / 4$, $-1 \le \im z , \im z' \le 1$ and $\im z \cdot \im z^{\prime} > 0$, we have
\begin{equation*}
{\mathcal K}_{\sigma} \leq \left\{ \begin{aligned}
&\mathrm{C}_{s , \varepsilon} \sigma^{\min ( s - \frac{1}{2} , \frac{3}{2} - s)} \vert z - z^{\prime} \vert^{s - \frac{1}{2} - \varepsilon} &&\text{ for } \sigma \in {\mathcal D} \setminus ( \Delta \cup \Delta^{\prime} ) ,  \\
&\mathrm{C}_{s , \varepsilon} \vert z - z^{\prime} \vert^{s - \frac{1}{2} - \varepsilon} &&\text{ for } \sigma \in \Delta \cup \Delta^{\prime} ,
\end{aligned} \right.
\end{equation*}
\end{lemma}

We first assume Lemma \ref{b7} and finish the proof of Theorem \ref{b8}. Using $1/2 < s < 3/2$ and that the cardinals of $\Delta$ and $\Delta^{\prime}$ are uniformly bounded with respect to $z$, $z^{\prime}$, \eqref{b1} gives
\begin{align*}
{\mathcal K} &\leq \sum_{\sigma \in {\mathcal D} \setminus ( \Delta \cup \Delta^{\prime} )} {\mathcal K}_{\sigma} + \sum_{\sigma \in \Delta \cup \Delta^{\prime}} {\mathcal K}_{\sigma} + \mathrm{C} \vert z - z^{\prime} \vert  \\
&\leq \sum_{\sigma \in {\mathcal D} \setminus ( \Delta \cup \Delta^{\prime} )} \mathrm{C}_{s , \varepsilon} \sigma^{\min ( s - \frac{1}{2} , \frac{3}{2} - s)} \vert z - z^{\prime} \vert^{s - \frac{1}{2} - \varepsilon} + ( \# \Delta + \# \Delta^{\prime} ) \mathrm{C}_{s , \varepsilon} \vert z - z^{\prime} \vert^{s - \frac{1}{2} - \varepsilon} + \mathrm{C} \vert z - z^{\prime} \vert  \\
&\leq \mathrm{C}_{s , \varepsilon} \vert z - z^{\prime} \vert^{s - \frac{1}{2}-\varepsilon} ,
\end{align*}
which is the required H\"{o}lder regularity of the resolvent from \eqref{b4}.

\begin{proof}[Proof of Lemma \ref{b7}] We distinguish between different cases.

Assume first that $\sigma \in {\mathcal D} \setminus ( \Delta \cup \Delta^{\prime} )$. On one hand, Corollary \ref{c15} and \eqref{c6} give
\begin{align}
{\mathcal K}_{\sigma} \leq{}& \big( \big\Vert ( H_\alpha - z )^{-1} \varphi_{\sigma} ( H_\alpha - E_\alpha ) \big\Vert + \big\Vert ( H_\alpha - z^{\prime} )^{-1} \varphi_{\sigma} ( H_\alpha - E_\alpha ) \big\Vert \big)  \nonumber \\
&\times \big\Vert \< X \>^{- s} \widetilde{\varphi}_{\sigma} ( H_\alpha - E_\alpha ) \big\Vert^2 \nonumber  \\
\leq{}& \mathrm{C} ( \sigma^{-1} + \sigma^{-1} ) \sigma^{2 \min ( 1 , s )} \leq \mathrm{C} \sigma^{\min ( 1 , 2 s -1 )} . \label{b2}
\end{align}
On the other hand, the resolvent identity yields
\begin{align}
{\mathcal K}_{\sigma} &\leq \vert z - z^{\prime} \vert \big\Vert \< X \>^{- s} ( H_\alpha - z )^{-1} ( H_\alpha - z^{\prime} )^{-1} \varphi_{\sigma} ( H_\alpha - E_\alpha ) \< X \>^{- s} \big\Vert   \nonumber  \\
&\leq \vert z - z^{\prime} \vert \big\Vert ( H_\alpha - z )^{-1} ( H_\alpha - z^{\prime} )^{-1} \varphi_{\sigma} ( H_\alpha - E_\alpha ) \big\Vert \big\Vert \< X \>^{- s} \widetilde{\varphi}_{\sigma} ( H_\alpha - E_\alpha ) \big\Vert^2 \nonumber  \\
&\leq \mathrm{C} \vert z - z^{\prime} \vert \sigma^{-2} \sigma^{2 \min ( 1 , s )} \leq \mathrm{C} \sigma^{\min ( 0 , 2 s - 2 )} \vert z - z^{\prime} \vert . \label{b3}
\end{align}
Thus, combining \eqref{b2} and \eqref{b3}, we get
\begin{equation*}
{\mathcal K}_{\sigma} \leq \mathrm{C} \sigma^{( \frac{3}{2} - s) \min ( 1 , 2 s -1 )} \sigma^{(s - \frac{1}{2} ) \min ( 0 , 2 s - 2 )} \vert z - z^{\prime} \vert^{s - \frac{1}{2}} \leq \mathrm{C} \sigma^{\min ( s - \frac{1}{2} , \frac{3}{2} - s)} \vert z - z^{\prime} \vert^{s - \frac{1}{2}} ,
\end{equation*}
and the first estimate of Lemma \ref{b7} follows.

We now assume that $\sigma \in \Delta \cup \Delta^{\prime}$ and $\Delta \cap \Delta^{\prime} = \emptyset$. If $\sigma \in \Delta$, Proposition \ref{c8} with $n =1$, Corollary \ref{c15}, Corollary \ref{c16}, Corollary \ref{c9} and the proof of \eqref{b2} imply
\begin{align*}
{\mathcal K}_{\sigma} \leq{}& \big\Vert \< X \>^{- s} \widetilde{\varphi}_{\sigma} ( H_\alpha - E_\alpha ) \< B^\sigma \>^{s} \big\Vert^{2} \big\Vert \< B^\sigma \>^{- s} ( H_\alpha - z )^{-1} \< B^\sigma \>^{- s} \big\Vert  \\
&+ \big\Vert ( H_\alpha - z^{\prime} )^{-1} \varphi_{\sigma} ( H_\alpha - E_\alpha ) \big\Vert \big\Vert \< X \>^{- s} \widetilde{\varphi}_{\sigma} ( H_\alpha - E_\alpha ) \big\Vert^2   \\
\leq{}& \mathrm{C}_{s , \delta} \sigma^{2 \min ( 1 , s ) - 2 \delta} \sigma^{- 1} + \mathrm{C} \sigma^{-1} \sigma^{2 \min ( 1 , s )} \leq \mathrm{C}_{s , \delta} \sigma^{\min ( 1 , 2 s -1 ) - 2 \delta} ,
\end{align*}
for all $\delta > 0$. Furthermore, since $\Delta \cap \Delta^{\prime} = \emptyset$, we have $\sigma \leq \mathrm{C} \vert z  - z^{\prime} \vert$ and the last equation becomes
\begin{equation*}
{\mathcal K}_{\sigma} \leq \mathrm{C}_{s , \delta} \sigma^{\min ( s - \frac{1}{2} , \frac{3}{2} - s) + \varepsilon - 2 \delta} \vert z - z^{\prime} \vert^{s - \frac{1}{2} - \varepsilon} = \mathrm{C}_{s , \varepsilon} \sigma^{\min ( s - \frac{1}{2} , \frac{3}{2} - s)} \vert z - z^{\prime} \vert^{s - \frac{1}{2} - \varepsilon} ,
\end{equation*}
with $\delta = \varepsilon /2$. For $\sigma \in \Delta^{\prime}$, ${\mathcal K}_{\sigma}$ satisfies the same estimate if $\Delta \cap \Delta^{\prime} = \emptyset$.

It remains to study $\sigma \in \Delta \cup \Delta^{\prime}$ under the condition $\Delta \cap \Delta^{\prime} \neq \emptyset$. We denote $I : = [ 1/M^{3} , M^{3} ] \subset ( 0 , + \infty )$. For $\sigma \in \Delta \cup \Delta^{\prime}$ with $\Delta \cap \Delta^{\prime} \neq \emptyset$, we have $\re z , \re z^{\prime} \in E_{\alpha} + \sigma I$. Proposition \ref{c8} with $n=2$ then gives
\begin{equation*}
\sup_{\fract{\re w \in E_{\alpha} + \sigma I}{\im w \neq 0}} \big\Vert \< B^{\sigma} \>^{- \frac{3}{2} - \varepsilon} ( H_\alpha - w )^{- 2} \< B^{\sigma} \>^{- \frac{3}{2} - \varepsilon} \big\Vert \leq \mathrm{C}_{\varepsilon} \sigma^{- 2} .
\end{equation*}
In particular, since $\re z , \re z^{\prime} \in E_{\alpha} + \sigma I$ and $\im z \cdot \im z^{\prime} > 0$, the mean-value theorem implies
\begin{equation*}
\big\Vert \< B^{\sigma} \>^{- \frac{3}{2} - \varepsilon} \big( ( H_\alpha - z )^{- 1} - ( H_{\alpha} - z^{\prime} )^{-1} \big) \< B^{\sigma} \>^{- \frac{3}{2} - \varepsilon} \big\Vert \leq \mathrm{C}_{\varepsilon} \sigma^{- 2} \vert z - z^{\prime} \vert .
\end{equation*}
On the other hand, combining Corollary \ref{c15}, Corollary \ref{c16} and Corollary \ref{c9} with an interpolation argument, we get
\begin{equation*}
\big\Vert \< X \>^{- \frac{3}{2} - \varepsilon} \chi_{\sigma} ( H_\alpha - E_\alpha ) \< B^\sigma \>^{\frac{3}{2} + \varepsilon} \big\Vert \leq \mathrm{C}_{\varepsilon , \delta} \sigma^{1 - \delta} ,
\end{equation*}
for all $\chi \in \mathrm{C}^{\infty}_{0} ( ( 0 , 1 ) )$ and $\delta > 0$. Then, the last two estimates yield
\begin{align*}
\big\Vert \< X \>^{- \frac{3}{2} - \varepsilon} \big( ( H_\alpha & - z )^{- 1} - ( H_{\alpha} - z^{\prime} )^{-1} \big) \varphi_{\sigma} ( H_\alpha - E_\alpha ) \< X \>^{- \frac{3}{2} - \varepsilon} \big\Vert    \\
\leq{}& \big\Vert \< B^{\sigma} \>^{- \frac{3}{2} - \varepsilon} \big( ( H_\alpha - z )^{- 1} - ( H_{\alpha} - z^{\prime} )^{-1} \big) \< B^{\sigma} \>^{- \frac{3}{2} - \varepsilon} \big\Vert  \\
&\times \big\Vert \< X \>^{- \frac{3}{2} - \varepsilon} \widetilde{\varphi}_{\sigma} ( H_\alpha - E_\alpha ) \< B^\sigma \>^{\frac{3}{2} + \varepsilon} \big\Vert   \big\Vert \< X \>^{- \frac{3}{2} - \varepsilon} \varphi_{\sigma} ( H_\alpha - E_\alpha ) \< B^\sigma \>^{\frac{3}{2} + \varepsilon} \big\Vert \\
\leq{}& \mathrm{C}_{\varepsilon , \delta} \sigma^{- 2 \delta} \vert z - z^{\prime} \vert .
\end{align*}
Moreover, from Lemma \ref{c7} with $s = 1/2 + \varepsilon$, we have
\begin{align*}
\big\Vert \< X \>^{- \frac{1}{2} - \varepsilon} \big( ( H_\alpha & - z )^{- 1} - ( H_{\alpha} - z^{\prime} )^{-1} \big) \varphi_{\sigma} ( H_\alpha - E_\alpha ) \< X \>^{- \frac{1}{2} - \varepsilon} \big\Vert    \\
\leq{}& \big\Vert \< X \>^{- \frac{1}{2} - \varepsilon} ( H_\alpha - z )^{- 1} \varphi_{\sigma} ( H_\alpha - E_\alpha ) \< X \>^{- \frac{1}{2} - \varepsilon} \big\Vert     \\
&+ \big\Vert \< X \>^{- \frac{1}{2} - \varepsilon} ( H_{\alpha} - z^{\prime} )^{-1} \varphi_{\sigma} ( H_\alpha - E_\alpha ) \< X \>^{- \frac{1}{2} - \varepsilon} \big\Vert    \\
\leq{}& \mathrm{C}_{\varepsilon , \delta} \sigma^{2 \varepsilon - \delta } .
\end{align*}
Then, an interpolation between the last two estimates implies, for $\varepsilon$ small enough,
\begin{equation*}
{\mathcal K}_{\sigma} \leq \mathrm{C}_{s,\varepsilon , \delta} \sigma^{\varepsilon ( 3 - 2 s + 2 \varepsilon + \delta ) - \delta (s + \frac{1}{2} )} \vert z - z^{\prime} \vert^{s - \frac{1}{2} - \varepsilon} \leq \mathrm{C}_{s,\varepsilon} \vert z - z^{\prime} \vert^{s - \frac{1}{2} - \varepsilon},
\end{equation*}
for $s \in ( 1/2 , 3/2 )$ and $\delta \ll \varepsilon$. This finishes the proof of the lemma.
\end{proof}

\subsection{Proof of the local decay}

We finally prove Theorem \ref{c1}. Since the assertion is clear for $s =0$, we can assume that $0 < s < 2$. Let $\varphi , \widetilde{\varphi} \in \mathrm{C}_0^\infty( (0, 1 ) ; \mathbb{R} )$ be as in \eqref{c3}. Then,
\begin{equation} \label{a24}
\forall x \in \supp ( \chi ( \cdot + E_{\alpha} ) ) , \qquad  \mathds{1}_{\{ 0 \}} (x) + \sum_{\sigma \in {\mathcal D}} \varphi_\sigma (x) = 1 .
\end{equation}
From Corollary \ref{c15}, Corollary \ref{c16} and Corollary \ref{c9}, we have
\begin{equation*}
\big\Vert \< X \>^{- s} \widetilde{\varphi}_{\sigma} ( H_{\alpha} - E_{\alpha} ) \langle B^{\sigma} \rangle^{s} \big\Vert \leq \left\{ \begin{aligned}
&\mathrm{C} &&\text{for } s =0 ,  \\
&\mathrm{C}_{\delta} \sigma^{1 - \delta} &&\text{for } s =1 , 2 ,
\end{aligned} \right.
\end{equation*}
for all $\delta > 0$. Therefore, an interpolation argument gives
\begin{equation} \label{a25}
\big\Vert \< X \>^{- s} \widetilde{\varphi}_{\sigma} ( H_{\alpha} - E_{\alpha} ) \langle B^{\sigma} \rangle^{s} \big\Vert \leq \mathrm{C}_{\delta} \sigma^{\min ( 1 , s ) - \delta} ,
\end{equation}
for all $s \in [0 ,2]$. Now, Remark \ref{g8} implies that $(B^{\sigma} )^{n} \chi ( H_{\alpha} ) \langle B^{\sigma} \rangle^{-n}$ is a uniformly bounded operator for all $n \in \mathbb{N} \cup \{ 0 \}$. Therefore, an interpolation argument gives that, for all $s \geq 0$, there exists $\mathrm{C}_{s , \chi} >0$ such that
\begin{align} \label{a26}
\big\Vert \langle B^{\sigma} \rangle^{s} \chi ( H_{\alpha} ) \langle B^{\sigma} \rangle^{- s} \big\Vert \leq \mathrm{C}_{s , \chi} .
\end{align}
Thus, using \eqref{a25} and \eqref{a26}, Proposition \ref{c14} gives
\begin{align}
\big\Vert \< X \>^{- s} e^{- \ii t H_{\alpha}} & \varphi_{\sigma} ( H_{\alpha} - E_{\alpha} ) \chi ( H_{\alpha} ) \< X \>^{- s} \big\Vert  \nonumber  \\
\leq{}& \big\Vert \< X \>^{- s} \widetilde{\varphi}_{\sigma} ( H_{\alpha} - E_{\alpha} ) \langle B^{\sigma} \rangle^{s} \big\Vert    \big\Vert \langle B^{\sigma} \rangle^{- s} e^{- \ii t H_{\alpha}} \varphi_{\sigma} ( H_{\alpha} - E_{\alpha} ) \langle B^{\sigma} \rangle^{s} \big\Vert  \nonumber    \\
&\times  \big\Vert \langle B^{\sigma} \rangle^{s} \chi ( H_{\alpha} ) \langle B^{\sigma} \rangle^{-s} \big\Vert  \big\Vert \langle B^{\sigma} \rangle^{s}  \widetilde{\varphi}_{\sigma} ( H_{\alpha} - E_{\alpha} ) \< X \>^{- s} \big\Vert   \nonumber \\
\leq{}& \mathrm{C}_{s , \delta , \chi} \sigma^{\min ( 1 , s ) - \delta} \langle t \sigma \rangle^{- s} \sigma^{\min ( 1 , s ) - \delta}  \leq \mathrm{C}_{s , \delta , \chi} \sigma^{\min (2 -s , s ) - 2 \delta} \langle t \rangle^{-s} .  \label{a23}
\end{align}
Eventually, \eqref{a24} implies
\begin{align}
\big\Vert \< X \>^{- s} e^{- \ii t H_{\alpha}} \chi ( H_{\alpha} & ) \< X \>^{- s} - \< X \>^{- s} e^{- \ii t E_{\alpha}} \chi ( E_{\alpha} ) \Pi_{\alpha} \< X \>^{- s} \big\Vert  \nonumber \\
&\leq \sum_{\sigma \in {\mathcal D}} \big\Vert \< X \>^{- s} e^{- \ii t H_{\alpha}} \varphi_{\sigma} ( H_{\alpha} - E_{\alpha} ) \chi ( H_{\alpha} ) \< X \>^{- s} \big\Vert  \nonumber \\
&\leq \mathrm{C}_{s , \delta , \chi} \langle t \rangle^{-s} \sum_{\sigma \in {\mathcal D}} \sigma^{\min (2 -s , s ) - 2 \delta} \leq \mathrm{C}_{s , \chi} \langle t \rangle^{-s} ,
\end{align}
since $\min (2 - s , s ) > 0$ for $0 < s < 2$. This finishes the proof of Theorem \ref{c1}.

\appendix

\section{Properties and technicalities}\label{d9}

In this appendix, we collect a few properties regarding the infrared decomposition and the infrared cutoff Hamiltonian which were used in Subsection \ref{d36}. The notations are the ones of Subsection \ref{d36}. 
Moreover, for $f: \mathbb{R}^3 \times \{1,2\} \mapsto \mathbb{C}$ and $\sigma > 0$, we define
\begin{equation}
f^\sigma( k,\lambda ) = f( k,\lambda ) \mathds{1}_{ |k| \le \sigma }(k),
\end{equation}
and, similarly, we set
\begin{equation}
H_f^\sigma = \sum_{\lambda=1,2} \int_{ |k| \le \sigma } |k| a^*_\lambda( k ) a_\lambda( k ) \d k.
\end{equation}
Observe that $H_f^\sigma = \mathds{1}_{ \mathcal{F}_{\ge\sigma} } \otimes \d \Gamma( |k| )$. We begin with recalling the following standard lemma.

\begin{lemma}\sl \label{d10}
Let $f \in \mathrm{L}^2( \mathbb{R}^3 \times \{1,2\} )$ be such that $(k,\lambda) \mapsto |k|^{-1/2} f(k,\lambda) \in \mathrm{L}^2( \mathbb{R}^3 \times \{ 1,2 \} )$. Then, for any $\sigma>0$ and $\rho > 0$, the operators $a(f^\sigma) ( H_f^\sigma + \rho )^{-1/2}$ and $a^*(f^\sigma) ( H_f^\sigma + \rho )^{-1/2}$ extend to bounded operators on $\mathcal{F}$ satisfying
\begin{align*}
\big\Vert a(f^\sigma) ( H_f^\sigma + \rho )^{- \frac{1}{2} } \big\Vert & \leq \big\Vert |k|^{-\frac{1}{2}} f^\sigma \big\Vert ,    \\
\big\Vert a^*(f^\sigma) ( H_f^\sigma + \rho )^{- \frac{1}{2} } \big\Vert & \leq \big\Vert |k|^{-\frac{1}{2}} f^\sigma \big\Vert + \rho^{-\frac{1}{2}} \Vert f^\sigma \Vert .
\end{align*}
Let in addition $g \in \mathrm{L}^2( \mathbb{R}^3 \times \{1,2\} )$ be such that $(k,\lambda) \mapsto |k|^{-1/2} g(k,\lambda) \in \mathrm{L}^2( \mathbb{R}^3 \times \{ 1,2 \} )$. Then we have
\begin{align*}
\big\Vert a(f^\sigma) a(g^\sigma) ( H_f^\sigma + \rho )^{- 1 } \big\Vert & \leq \big\Vert |k|^{-\frac{1}{2}} f^\sigma \big\Vert \big\Vert |k|^{-\frac{1}{2}} g^\sigma \big\Vert ,     \\
\big\Vert a^*(f^\sigma) a(g^\sigma) ( H_f^\sigma + \rho )^{- 1 } \big\Vert & \leq \big( \big\Vert |k|^{-\frac{1}{2}} f^\sigma \big\Vert + \rho^{-\frac{1}{2}} \Vert f^\sigma \Vert \big) \big\Vert |k|^{-\frac{1}{2}} g^\sigma \big\Vert ,        \\
\big\Vert a^*(f^\sigma) a^*(g^\sigma) ( H_f^\sigma + \rho )^{- 1 } \big\Vert & \leq \big( \big\Vert |k|^{-\frac{1}{2}} f^\sigma \big\Vert + \rho^{-\frac{1}{2}} \Vert f^\sigma \Vert \big) \big( \big\Vert |k|^{-\frac{1}{2}} g^\sigma \big\Vert + \rho^{-\frac{1}{2}} \Vert g^\sigma \Vert \big) .
\end{align*}
\end{lemma}

The following lemma is proven in \cite{FGS1}.

\begin{lemma}[{\cite[Lemma 22]{FGS1}}]\sl \label{e1}
There exist $\alpha_c>0$ and $\mathrm{C}>0$ such that, for all $0 \leq \alpha \leq \alpha_c$ and $0 < \sigma \leq e_{\mathrm{gap}}/2$, 
\begin{equation*}
| E_\alpha - E_{\alpha,\sigma} | \leq \mathrm{C} \alpha^{\frac{3}{2}} \sigma^2.
\end{equation*}
\end{lemma}

Using Lemma \ref{d10} and Lemma \ref{e1}, we now establish the following lemma which will be useful in the sequel.

\begin{lemma}\sl \label{e2}
There exist $\alpha_c>0$ and $\mathrm{C}> 0$ such that, for all $0 \leq \alpha \leq \alpha_c$ and $0 < \sigma \leq e_{\mathrm{gap}}/2$, 
\begin{equation*}
\big\Vert ( \mathds{1}_{ \mathcal{H}_{\ge\sigma} } \otimes H_f ) \Phi \big\Vert \leq \mathrm{C} \Vert ( H_\alpha - E_\alpha ) \Phi \Vert + \mathrm{C} \sigma \| \Phi \| ,
\end{equation*}
for all $\Phi \in D(H_0)$.
\end{lemma}

\begin{proof}
Let 
\begin{equation} \label{e6}
W_{\alpha,\sigma} : = H_\alpha - H_{\alpha,\sigma} = 2 \alpha^{\frac{3}{2}} A^{\le\sigma}( \alpha x ) \cdot \big( p + \alpha^{\frac{3}{2}} A_{\geq \sigma}(\alpha x) \big) + \alpha^3 \big( A^{\leq \sigma}( \alpha x ) \big)^2.
\end{equation}
Since $H_{\alpha,\sigma} = K_{\alpha,\ge\sigma} \otimes \mathds{1}_{ \mathcal{F}^{\le\sigma} } + \mathds{1}_{ \mathcal{H}_{\ge\sigma} } \otimes H_f$, we have
\begin{equation}\label{e3}
\big\Vert ( \mathds{1}_{ \mathcal{H}_{\ge\sigma} } \otimes H_f ) \Phi \big\Vert \leq \Vert ( H_{\alpha,\sigma} - E_{\alpha,\sigma} ) \Phi \Vert \le \| ( H_\alpha - E_\alpha ) \Phi \| + \| W_{\alpha,\sigma} \Phi \| + \mathrm{C} \alpha^{\frac{3}{2}} \sigma^2,
\end{equation}
where we used Lemma \ref{e1} in the last inequality. It follows from Lemma \ref{d10} that
\begin{equation*}
\big\Vert \big( A^{\leq \sigma}( \alpha x ) \big)^2 \big( \mathds{1}_{ \mathcal{H}_{\ge\sigma}} \otimes H_f + \sigma \big)^{-1} \big\Vert \leq \mathrm{C} \sigma.
\end{equation*}
Moreover, since $\| ( p + \alpha^{\frac{3}{2}} A_{\geq \sigma}(\alpha x) ) \Psi \| \le \mathrm{C} \| K_{0,\ge\sigma} \Psi \| + \mathrm{C} \| \Psi \|$ for all $\Psi \in D(K_{0,\ge\sigma})$, we have
\begin{align}
\big\Vert & A^{\le\sigma}( \alpha x ) \cdot \big( p + \alpha^{\frac{3}{2}} A_{\geq \sigma}(\alpha x) \big) \big( \mathds{1}_{ \mathcal{H}_{\ge\sigma} } \otimes H_f + \sigma \big)^{-\frac{1}{2}} \big( ( K_{0,\ge\sigma} - e_{1} + 1 ) \otimes \mathds{1}_{ \mathcal{F}^{\le\sigma} } \big)^{-\frac{1}{2}} \big\Vert \nonumber \\
& \leq \big\Vert A^{\le\sigma}( \alpha x ) \big( \mathds{1}_{ \mathcal{H}_{\ge\sigma} } \otimes H_f + \sigma \big)^{-\frac{1}{2}} \big\Vert \big\Vert \big( p + \alpha^{\frac{3}{2}} A_{\geq \sigma}(\alpha x) \big) \big( ( K_{0,\ge\sigma} - e_{1} + 1 ) \otimes \mathds{1}_{ \mathcal{F}^{\le\sigma} } \big)^{-\frac{1}{2}} \big\Vert \nonumber \\
& \leq \mathrm{C} \sigma^{\frac{1}{2}}.
\end{align}
Combining the preceding two estimates with \eqref{e6}, we obtain
\begin{equation}
\big\Vert W_{\alpha,\sigma} \big( \mathds{1}_{ \mathcal{H}_{\ge\sigma} } \otimes H_f + \sigma ( K_{0,\ge\sigma} - e_{1} + 1 ) \otimes \mathds{1}_{ \mathcal{F}^{ \leq \sigma } } \big)^{-1} \big\Vert \leq \mathrm{C} \alpha^{ \frac{3}{2} }. \label{e4}
\end{equation}
Since $\Vert ( K_{0,\ge\sigma} - e_{1} + 1 ) \Phi \Vert \leq \mathrm{C} \| ( H_\alpha - E_\alpha ) \Phi \| + \mathrm{C} \| \Phi \|$, we conclude from \eqref{e4} that
\begin{equation}\label{e5}
\big\Vert W_{\alpha,\sigma} \Phi \big\Vert \leq \mathrm{C} \alpha^{ \frac{3}{2} } \sigma \| ( H_\alpha - E_\alpha ) \Phi \| + \mathrm{C} \alpha^{ \frac{3}{2} } \sigma \| \Phi \| + \mathrm{C} \alpha^{ \frac{3}{2} } \big\| ( \mathds{1}_{ \mathcal{H}_{\ge\sigma} } \otimes H_f ) \Phi \big\Vert ,
\end{equation}
For $\alpha$ small enough, \eqref{e3} and \eqref{e5} imply the statement of the lemma.
\end{proof}

The next lemma is established in \cite{FGS1}. It is based on the fact that states with spectral support below the ionization thresholds decay exponentially in the electron position variable (see \cite{BFS,Gr}).

\begin{lemma}[{\cite[Lemma 17]{FGS1}}] \sl \label{e7}
For  all $\lambda < e_2$, there exists $\alpha_\lambda>0$ such that, for all $0 \leq \alpha \leq \alpha_\lambda$ and $n \in \mathbb{N} \cup \{ 0 \}$,
\begin{equation*}
\sup_{\sigma \geq 0} \big\Vert \langle x \rangle^n \mathds{1}_{ ( - \infty , \lambda ] }( H_{\alpha,\sigma} ) \big\Vert \leq \mathrm{C} ,
\end{equation*}
where $\mathrm{C}$ is a positive constant independent of $\sigma$.
\end{lemma}

We now give the following result that will be useful in the next appendix.

\begin{lemma}\sl \label{d7}
For all $n \in \mathbb{N} \cup \{ 0 \}$, there exists $\mathrm{C}_n > 0$ such that, for all $\alpha \geq 0$, $0 \leq \sigma \leq e_{\mathrm{gap}}/2$, $\tau \geq 0$ and $z \in \mathbb{C}$, $0 < \pm \im z \leq 1$, the operator $\langle \sigma x \rangle^{-n} ( H_{\alpha,\tau} - z )^{-1} \langle \sigma x \rangle^n$ defined on $D(\langle x \rangle^n)$ extends by continuity to a bounded operator on $\mathcal{H}$  satisfying
\begin{align}\label{e8}
\big\Vert \langle \sigma x \rangle^{-n} ( H_{\alpha,\tau} - z )^{-1} \langle \sigma x \rangle^n \big\Vert \leq \mathrm{C}_n \Big\< \frac{\sigma}{\vert \im z \vert} \Big\>^{n} \frac{1}{\vert \im z \vert} .
\end{align}
Moreover, $\langle \sigma x \rangle^{-n} ( H_{\alpha,\tau} - z )^{-1} \langle \sigma x \rangle^n (H_{\alpha,\tau}-z)$ defined on $D(H_0)$ extends by continuity to a bounded operator on $\mathcal{H}$ satisfying
\begin{align}\label{e9}
\big\Vert \langle \sigma x \rangle^{-n} ( H_{\alpha,\tau} - z )^{-1} \langle \sigma x \rangle^n (H_{\alpha,\tau}-z) \big\Vert \leq  \mathrm{C}_n \Big\< \frac{\sigma}{\vert \im z \vert} \Big\>^{n} .
\end{align}
\end{lemma}

\begin{proof}
We proceed by induction. For $n=0$, \eqref{e8} follows from the spectral theorem and \eqref{e9} is obvious. Now suppose that \eqref{e8}--\eqref{e9} hold for any $k \in \mathbb{N}\cup\{0\}$, $k \leq n$, where $n \in \mathbb{N}\cup\{0\}$. For any $\varepsilon > 0$, we can write
\begin{align}
\frac{ 1 }{ \langle \sigma x \rangle^{n+1} } ( H_{\alpha,\tau} - z )^{-1} & \frac{ \langle \sigma x \rangle^{n+1} }{ 1 + \varepsilon \langle \sigma x \rangle^{n+1} } = \frac{ 1 }{ 1 + \varepsilon \langle \sigma x \rangle^{n+1} } (H_{\alpha,\tau} - z )^{-1} \notag \\
 &- \frac{ 1 }{ \langle \sigma x \rangle^{n+1} } (H_{\alpha,\tau} - z)^{-1} \Big[ H_{\alpha,\tau} , \frac{ \langle \sigma x \rangle^{n+1} }{ 1 + \varepsilon \langle \sigma x \rangle^{n+1} } \Big] (H_{\alpha,\tau} - z)^{-1},
\end{align}
in the sense of quadratic forms on $\mathcal{H} \times \mathcal{H}$. We compute
\begin{align*}
\Big[ H_{\alpha,\tau} & , \frac{\langle \sigma x \rangle^{n+1}}{1 + \varepsilon \langle \sigma x \rangle^{n+1}} \Big] = - 2 \ii (n+1) \sigma \langle \sigma x \rangle^{n} \frac{ \sigma x }{ \langle \sigma x \rangle ( 1 + \varepsilon \langle \sigma x \rangle^{n+1} )^2 } \cdot \big( p + \alpha^{\frac{3}{2}} A_{\ge\tau}( \alpha x ) \big)    \\
&- \sigma^2 \langle \sigma x \rangle^{n-1} \Big ( \frac{ (n+1) }{ ( 1 + \varepsilon \langle \sigma x \rangle^{n+1} )^2 } - \frac{(n+1) (n+3) \sigma^2 x^2 }{ \langle \sigma x \rangle^2 ( 1 + \varepsilon \langle \sigma x \rangle^{n+1} )^2 } + \frac{ 2 (n+1)^2 \sigma^2 x^2 }{ \langle \sigma x \rangle^2 ( 1 + \varepsilon \langle \sigma x \rangle^{n+1} )^3 } \Big ),
\end{align*}
and
\begin{align*}
\Big[ H_{\alpha,\tau} & , \frac{\langle \sigma x \rangle^{n+1}}{1 + \varepsilon \langle \sigma x \rangle^{n+1}} \Big] = - 2 \ii (n+1) \sigma \langle \sigma x \rangle^{n} \big( p + \alpha^{\frac{3}{2}} A_{\ge\tau}( \alpha x ) \big) \cdot \frac{ \sigma x }{ \langle \sigma x \rangle ( 1 + \varepsilon \langle \sigma x \rangle^{n+1} )^2 }  \\
&+ \sigma^2 \langle \sigma x \rangle^{n-1} \Big ( \frac{ (n+1) }{ ( 1 + \varepsilon \langle \sigma x \rangle^{n+1} )^2 } - \frac{3 (n+1)^{2} \sigma^2 x^2 }{ \langle \sigma x \rangle^2 ( 1 + \varepsilon \langle \sigma x \rangle^{n+1} )^2 } + \frac{ 2 (n+1)^2 \sigma^2 x^2 }{ \langle \sigma x \rangle^2 ( 1 + \varepsilon \langle \sigma x \rangle^{n+1} )^3 } \Big ),
\end{align*}
in the sense of quadratic forms on $D( H_0 ) \times D( H_0 )$. Combining the induction hypothesis with the fact that
\begin{equation*}
\big\Vert (H_{\alpha,\tau} - z)^{-1} \big( p + \alpha^{\frac{3}{2}} A_{\ge\tau}( \alpha x ) \big) \big\Vert \leq \frac{ \mathrm{C} }{ | \im z| } ,
\end{equation*}
next letting $\varepsilon \to 0$, it is seen that \eqref{e8}--\eqref{e9} hold with $n+1$ substituted for $n$, which concludes the proof of the lemma.
\end{proof}

To conclude this section, we recall

\begin{proposition}[{\cite[Proposition 7]{FGS1}}]\sl \label{c10}
There exists $\alpha_c>0$ such that, for all function $\varphi \in \mathrm{C}_0^\infty( (-\infty, 1 ) ; \mathbb{R} )$, there exists $\mathrm{C}_\varphi > 0$ such that, for all $0 \leq \alpha \leq \alpha_c$ and $0 < \sigma \leq e_{\mathrm{gap}} / 2$,
\begin{equation*}
\big\Vert \varphi_\sigma( H_\alpha - E_\alpha ) - \varphi_\sigma( H_{\alpha,\sigma} - E_{\alpha,\sigma} ) \big\Vert \leq \mathrm{C}_\varphi \alpha^{ \frac{3}{2} } \sigma.
\end{equation*}
\end{proposition}

\section{Uniform multiple commutators estimates}\label{c12}

We begin with recalling the following lemma.

\begin{lemma}\emph{(\cite[Proposition 9]{FGS1})}\sl \label{e16}
For all $s \in \mathbb{R}$ and $\sigma>0$, $e^{\ii s B^\sigma} D( H_0 ) \subset D(H_0)$.
\end{lemma}

For all $s \in \mathbb{R} \setminus \{0\}$, let $B^\sigma_s := ( e^{\ii s B^\sigma} - 1 ) / s$. The preceding lemma shows that the multiple commutators $\ad^{n}_{\ii B_s^\sigma} ( H_\alpha )$ are well-defined on $D(H_0)$ for all $n \in \mathbb{N}\cup\{0\}$. For $s=0$, we set $B^\sigma_0 := B^\sigma$. Mimicking the proof of \cite[Proposition 10]{FGS1}, one can verify the following lemma.

\begin{lemma}\sl \label{d8}
There exists $\alpha_c>0$ such that, for all $0 \leq \alpha \leq \alpha_c$, $\sigma \geq 0$, $n \in \mathbb{N}\cup\{0\}$ and $\Psi \in \mathcal{H}$, we have
\begin{align}
\lim_{s\to0} \langle x \rangle^{-n} \ad^{n}_{\ii B_s^\sigma} & ( H_\alpha ) ( H_0 + \ii )^{-1} \Psi = \langle x \rangle^{-n} \d \Gamma \big( \eta_{\sigma}^{n} (k) |k| \big) ( H_0 + \ii )^{-1} \Psi  \phantom{ \sum^n }  \notag \\
&+ (-1)^n \sum_{\fract{0 \leq j_1,j_2 \leq n}{j_1+j_2 = n}} \langle x \rangle^{-n} \big( \Phi^{(j_1)} \cdot \Phi^{(j_2)} + \Phi^{(j_2)} \cdot \Phi^{(j_1)} \big) ( H_0 + \ii )^{-1} \Psi. \label{e17}
\end{align}
Here $\eta_{\sigma}^{n} (k) = \eta^{n} ( k / \sigma )$ with $\eta^{n} \in \mathrm{C}_0^\infty( \{ \vert k \vert \leq 1 \} )$ and we have set
\begin{align}
\Phi^{(0)} &: = p + \alpha^{\frac{3}{2}} \Phi( h( \alpha x ) )    \label{g9}  \\
\Phi^{(j)} &: = \alpha^{\frac{3}{2}} \Phi ( \ii^j (b^\sigma)^j h (\alpha x) ) , \quad j \geq 1 .  \label{g10}
\end{align}
Moreover,
\begin{equation}
\sup_{ |s| \leq 1 } \big\Vert \langle x \rangle^{-n} \ad^{n}_{\ii B_s^\sigma} ( H_\alpha ) ( H_0 + \ii )^{-1} \big\Vert \leq \mathrm{C}_n(\sigma) ,
\end{equation}
where $\mathrm{C}_n(\sigma)$ is a positive constant depending on $n$ and $\sigma$, and for $s=0$,
\begin{equation}
\ad^{n}_{\ii B^\sigma} ( H_\alpha ) := \d \Gamma \big( \eta_{\sigma}^{n} (k) |k| \big) + (-1)^n \sum_{\fract{0 \leq j_1,j_2 \leq n}{j_1+j_2 = n}} \big( \Phi^{(j_1)} \cdot \Phi^{(j_2)} + \Phi^{(j_2)} \cdot \Phi^{(j_1)} \big),
\end{equation}
as an operator in $\mathcal{B}( D(H_0) ; D(\langle x \rangle^{n})^* )$.
\end{lemma}

\begin{remark}\sl
Lemma \ref{e16} and Lemma \ref{d8} show that, for all $m \in \mathbb{N}\cup\{0\}$ and $|s|\le1$, the operators $\ad^n_{B^\sigma_s}(H_\alpha) \langle x \rangle^m (H_0+\ii )^{-1}$ are well-defined on $D(\langle x \rangle^{m})$. Commuting $\langle x \rangle^m$ with $\ad^{n}_{\ii B_s^\sigma} ( H_\alpha )$ in a way similar to what was done in the proof of Lemma \ref{d7}, it is not difficult to verify that for all $n,m \in \mathbb{N}\cup\{0\}$,  $\langle x \rangle^{-(n+m)} \ad^{n}_{\ii B_s^\sigma} ( H_\alpha ) \langle x \rangle^{m} (H_0+\ii )^{-1}$ extend by continuity to bounded operators on $\mathcal{H}$. Moreover, as in Lemma \ref{d8}, we have that
\begin{align}
\lim_{s\to0} \langle x \rangle^{-(n+m)} & \ad^{n}_{\ii B_s^\sigma} ( H_\alpha ) \langle x \rangle^{m} ( H_0 + \ii )^{-1} \Psi = \langle x \rangle^{-n} \d \Gamma \big( \eta_{\sigma}^{n} (k) |k| \big) ( H_0 + \ii )^{-1} \Psi  \phantom{ \sum^n }  \notag \\
+ (- & 1 )^n \sum_{\fract{0 \leq j_1,j_2 \leq n}{j_1+j_2 = n}} \langle x \rangle^{-(n+m)} \big( \Phi^{(j_1)} \cdot \Phi^{(j_2)} + \Phi^{(j_2)} \cdot \Phi^{(j_1)} \big) \langle x \rangle^{m} ( H_0 + \ii )^{-1} \Psi, \label{e18}
\end{align}
for all $\Psi \in \mathcal{H}$, and
\begin{equation}
\sup_{ |s| \leq 1 } \big\Vert \langle x \rangle^{-(n+m)} \ad^{n}_{\ii B_s^\sigma} ( H_\alpha ) \langle x \rangle^{m} ( H_0 + \ii )^{-1} \big\Vert \leq \mathrm{C}_{n,m}(\sigma).
\end{equation}
Similarly, commuting now $\langle x \rangle^{-n+m}$ with $\ad^{n}_{\ii B_s^\sigma} ( H_\alpha )$, one verifies that for all $n,m \in \mathbb{N}\cup\{0\}$,  $(H_0+\ii )^{-1} \langle x \rangle^{-(n+m)} \ad^{n}_{\ii B_s^\sigma} ( H_\alpha ) \langle x \rangle^{m}$ extend by continuity to bounded operators on $\mathcal{H}$ such that
\begin{align}
\lim_{s\to0} ( H_0 + & \ii )^{-1} \langle x \rangle^{-(n+m)} \ad^{n}_{\ii B_s^\sigma} ( H_\alpha ) \langle x \rangle^{m} \Psi = ( H_0 + \ii )^{-1} \langle x \rangle^{-n} \d \Gamma \big( \eta_{\sigma}^{n} (k) |k| \big) \Psi  \phantom{ \sum^n }  \notag \\
&+ (-1)^n \sum_{\fract{0 \leq j_1,j_2 \leq n}{j_1+j_2 = n}} ( H_0 + \ii )^{-1} \langle x \rangle^{-(n+m)} \big( \Phi^{(j_1)} \cdot \Phi^{(j_2)} + \Phi^{(j_2)} \cdot \Phi^{(j_1)} \big) \langle x \rangle^{m} \Psi, \label{e19}
\end{align}
for all $\Psi \in \mathcal{H}$, and
\begin{equation}
\sup_{ |s| \leq 1 } \big\Vert ( H_0 + \ii )^{-1} \langle x \rangle^{-(n+m)} \ad^{n}_{\ii B_s^\sigma} ( H_\alpha ) \langle x \rangle^{m} \big\Vert \leq \mathrm{C}_{n,m}(\sigma).
\end{equation}
\end{remark}

\begin{lemma}\sl \label{c19}
There exists $\alpha_c>0$ such that for all $0 \leq \alpha \leq \alpha_c$, $\sigma \geq 0$, $n,m \in \mathbb{N}\cup\{0\}$, $0<|s|\leq 1$ and $z \in \mathbb{C} \setminus \mathbb{R}$, the operators $\langle x \rangle^{-(n+m)} \ad^n_{\ii B^\sigma_s} ( (H_\alpha-z)^{-1} ) \langle x \rangle^m$ defined on $D(\langle x \rangle^m )$ extend by continuity to bounded operators on $\mathcal{H}$, and we have
\begin{align}
\lim_{s\to 0} \< x & \>^{-(n+m)} \ad^n_{\ii B^\sigma_s} \big( (H_\alpha-z)^{-1} \big) \langle x \rangle^m \Psi \notag \\
&= \sum_{\fract{1 \leq j_1, \dots  , j_n \leq n}{j_1+ \dots +j_n =n}} c_{j_1,\dots,j_n} \langle x \rangle^{-(n+m)} ( H_\alpha - z )^{-1} \langle x \rangle^{n+m} \notag \\
&\phantom{ = \sum_{\fract{1 \leq j_1, \dots  , j_n \leq n}{j_1+ \dots +j_n =n}}} \prod_{1\leq l \leq n} \langle x \rangle^{-t_{l-1}}  \ad^{j_l}_{\ii B^\sigma}(H_\alpha) \langle x \rangle^{t_{l}} \langle x \rangle^{-t_{l}} ( H_\alpha - z )^{-1} \langle x \rangle^{t_{l}} \Psi, \label{e20}
\end{align}
for any $\Psi \in \mathcal{H}$, where $t_0= n + m$, $t_l := n+m - \sum_{i=1}^l j_i$ for $l \geq 1$, and $c_{j_1,\dots,j_n}$ are explicitly computable integers. Moreover,
\begin{align}\label{e21}
\sup_{ |s|\leq 1} \big\Vert  \langle x \rangle^{-(n+m)} \ad^n_{\ii B^\sigma_s} \big( (H_\alpha-z)^{-1} \big) \langle x \rangle^m \big\Vert \leq \frac{ \mathrm{C}_{n,m}(\sigma) }{ | \im z | } P_{n,m} ( | \im z |^{-1} ),
\end{align}
where $\mathrm{C}_{n,m}(\sigma)$ is a positive constant depending on $n$, $m$ and $\sigma$, $P_{n,m}$ is a polynomial with positive coefficients and degree $n+m+\sum_{l=1}^n t_l$, and, for $s=0$, $ \langle x \rangle^{-(n+m)} \ad^n_{\ii B^\sigma} ( (H_\alpha-z)^{-1} ) \langle x \rangle^m$ is defined as the bounded operator appearing in the right hand side of \eqref{e20}.
\end{lemma}

\begin{proof}
Let us prove \eqref{e20}. A straightforward computation gives
\begin{equation}
\ad^n_{\ii B^\sigma_s} \big( (H_\alpha-z)^{-1} \big) = \sum_{\fract{1 \leq j_1, \dots  , j_n \leq n}{j_1+ \dots +j_n =n}} c_{j_1,\dots,j_n} ( H_\alpha - z )^{-1} \prod_{1\leq l \leq n} \big( \ad^{j_l}_{\ii B^\sigma_s}(H_\alpha) ( H_\alpha - z)^{-1} \big),
\end{equation}
for some explicitly computable integers $c_{j_1,\dots,j_n}$, where the right hand side is a well-defined bounded operator on $\mathcal{H}$ according to Lemma \ref{e16}. Thus, $\langle x \rangle^{-(n+m)} \ad^n_{\ii B^\sigma_s} ( (H_\alpha-z)^{-1} ) \langle x \rangle^m$ is equal to the right hand side of \eqref{e20} with $B^\sigma_s$ in place of $B^\sigma$, and it remains to justify the strong convergence. By Lemma \ref{d7} and Lemma \ref{d8}, for all $1\leq l \leq n$, the operators $\langle x \rangle^{-t_{l-1}}  \ad^{j_l}_{\ii B^\sigma_s}(H_\alpha) \langle x \rangle^{t_{l}} \langle x \rangle^{-t_{l}} ( H_\alpha - z )^{-1} \langle x \rangle^{t_{l}}$ strongly converge as $s\to 0$, and are uniformly bounded on $|s|\leq 1$ by a constant of the form given in the right hand side of \eqref{e21}. It follows that
\begin{align*}
\slim_{s\to0} \prod_{1\leq l \leq n} \langle x \rangle^{-t_{l-1}} \ad^{j_l}_{\ii B^\sigma_s} & (H_\alpha) \langle x \rangle^{t_{l}} \langle x \rangle^{-t_{l}} ( H_\alpha - z )^{-1} \langle x \rangle^{t_{l}}  \\
&= \prod_{1\leq l \leq n} \langle x \rangle^{-t_{l-1}}  \ad^{j_l}_{\ii B^\sigma}(H_\alpha) \langle x \rangle^{t_{l}} \langle x \rangle^{-t_{l}} ( H_\alpha - z )^{-1} \langle x \rangle^{t_{l}},
\end{align*}
and that \eqref{e21} holds, which concludes the proof of the lemma.
\end{proof}

\begin{lemma}\sl \label{e22}
There exists $\alpha_c>0$ such that, for all $0 \leq \alpha \leq \alpha_c$, $\sigma \geq 0$, $n \in \mathbb{N}\cup\{0\}$ and $\varphi \in \mathrm{C}_0^\infty( ( - \infty , E_{\alpha} + e_{\mathrm{gap}} /2 ) ; \mathbb{R} )$, the quadratic forms $\ad^n_{B^\sigma}( \varphi ( H_\alpha ) )$ defined iteratively on $D( B^\sigma )$ extend by continuity to bounded quadratic forms on $\mathcal{H}$. The associated bounded operators on $\mathcal{H}$ are denoted by the same symbols. They satisfy
\begin{equation*}
\ad^n_{\ii B^\sigma} ( \varphi ( H_\alpha ) ) = \slim_{s\to0} \ad^n_{\ii B^\sigma_s}( \varphi ( H_\alpha ) ) .
\end{equation*}
\end{lemma}

\begin{proof}
We prove the lemma by induction. For $n=0$, there is nothing to prove. Assume that the statement of the lemma is established with $n-1$ substituted for $n$, where $n \in \mathbb{N}$. For any $\Phi , \Psi \in D( B^\sigma )$, we have that
\begin{align}
\big\langle \Phi , \ad^{n}_{\ii B^\sigma} ( \varphi ( H_\alpha ) ) \Psi \big\rangle :=& \big\langle \Phi , \big[ \ad^{n-1}_{\ii B^\sigma} ( \varphi ( H_\alpha ) ) , \ii B^\sigma \big] \Psi \big\rangle \notag \\
=&\lim_{s\to0} \big\langle \Phi , \big[ \ad^{n-1}_{\ii B^\sigma_s} ( \varphi ( H_\alpha ) ) , \ii B^\sigma_s \big] \Psi \big\rangle \notag \\
=& \lim_{s\to0} \big\langle \Phi , \ad^{n}_{\ii B^\sigma_s} ( \varphi ( H_\alpha ) ) \Psi \big\rangle. \label{e23}
\end{align}
Set $\varphi_0 := \varphi$ and consider $\varphi_1, \dots , \varphi_n \in \mathrm{C}_0^\infty( ( - \infty , E_{\alpha} + e_{\mathrm{gap}} / 2 ) ; \mathbb{R} )$ such that $\varphi_l \varphi_{l+1} = \varphi_l$ for any $0 \leq l \leq n$. Leibniz' rule gives
\begin{align}\label{e24}
\ad^{n}_{\ii B^\sigma_s} ( \varphi ( H_\alpha ) ) = \sum_{\fract{0 \leq j_0 , \dots  , j_n \leq n}{j_0 + \dots +j_n =n}} c_{j_0 ,\dots,j_n} \prod_{0\leq l \leq n} \ad^{j_l}_{\ii B^\sigma_s} ( \varphi_l ( H_\alpha ) ) ,
\end{align}
for some explicitly computable integers $c_{j_0 ,\dots,j_n}$. For each term $\prod_{0\leq l \leq n} \ad^{j_l}_{\ii B^\sigma_s} ( \varphi_l ( H_\alpha ) )$ appearing in the sum, there is at least one $l_0 \in \{ 0,\dots,n \}$ such that $j_{l_0}=0$. Given this $l_0$, we write
\begin{align}
\prod_{0\leq l \leq n} \ad^{j_l}_{\ii B^\sigma_s} ( \varphi_l ( H_\alpha ) ) &= \prod_{0\leq l \leq l_0 - 1} \big(\langle x \rangle^{s_l} \ad^{j_l}_{\ii B^\sigma_s} ( \varphi_l ( H_\alpha ) ) \langle x \rangle^{-s_{l+1}} \big) \notag \\
& \big( \langle x \rangle^{s_{l_0}} \varphi_{l_0}( H_\alpha ) \langle x \rangle^{\widetilde{s}_{l_0}} \big) \prod_{l_0+1\leq l \leq n} \big( \langle x \rangle^{-\widetilde{s}_{l-1}} \ad^{j_l}_{\ii B^\sigma_s} ( \varphi_l ( H_\alpha ) ) \langle x \rangle^{\widetilde{s}_l} \big), \label{e25}
\end{align}
where $s_0=0$, $s_l = \sum_{i=0}^{l-1} j_i$, $\widetilde{s}_l = \sum_{i=l_0+1}^n j_i$, and $\widetilde{s}_n = 0$. From Lemma \ref{e7}, the operator $\langle x \rangle^{s_{l_0}} \varphi_{l_0}( H_\alpha ) \langle x \rangle^{\widetilde{s}_{l_0}}$ is bounded. Let $\widehat{\varphi}_l \in \mathrm{C}_0^{\infty}( \mathbb{C} )$ denote an almost analytic extension of $\varphi_l$ satisfying $| \partial_{\bar z}\widehat{\varphi}_l ( z ) | \leq \mathrm{C}_{\varphi_l}^{(m)} |y|^m$ where $m\in\mathbb{N}$ is fixed sufficiently large, and where $z = x + \ii y$ and $\partial_{\bar z} = \partial_x + \ii \partial_y$. Then by Lemma \ref{c19}, we can write
\begin{align}
\slim_{s\to 0} \big( \langle x \rangle^{s_l} & \ad^{j_l}_{\ii B^\sigma_s} ( \varphi_l ( H_\alpha ) ) \langle x \rangle^{-s_{l+1}} \big) \notag \\
&= - \frac{1}{\pi} \int_{ \mathbb{R}^2 } \partial_{\bar z} \widehat{\varphi}_l(z) \slim_{s\to0} \big( \langle x \rangle^{s_l} \ad^{j_l}_{\ii B^\sigma_s} \big( ( H_\alpha - z )^{-1} \big) \langle x \rangle^{-s_{l+1}} \big) \, \d x \, \d y, \label{e26}
\end{align}
where the strong convergence holds on $\mathcal{H}$. Moreover using \eqref{e21} and the properties of $\widehat{\varphi}_l$, we obtain
\begin{equation*}
\sup_{0 < |s|\leq 1} \big\Vert \langle x \rangle^{s_l} \ad^{j_l}_{\ii B^\sigma_s} ( \varphi_l ( H_\alpha ) ) \langle x \rangle^{-s_{l+1}} \big\Vert < \infty,
\end{equation*}
for any $0 \leq l \leq l_0-1$. The same holds for $\langle x \rangle^{-\widetilde{s}_{l-1}} \ad^{j_l}_{\ii B^\sigma_s} ( \varphi_l ( H_\alpha ) ) \langle x \rangle^{\widetilde{s}_l}$ in the case where $l_0+1 \leq l \leq n$. It follows that $\prod_{0\leq l \leq n} \ad^{j_l}_{\ii B^\sigma_s} ( \varphi_l ( H_\alpha ) )$ strongly converges as $s \to 0$ and is uniformly bounded on $|s| \leq 1$. Together with \eqref{e23}, this shows that
\begin{align}
\big| \big\langle \Phi , \ad^{n}_{\ii B^\sigma} ( \varphi ( H_\alpha ) ) \Psi \big\rangle \big| \leq \mathrm{C} \Vert \Phi \Vert \Vert \Psi \Vert ,
\end{align}
and that  $\ad^{n}_{\ii B^\sigma} ( \varphi ( H_\alpha ) ) = \slim_{s\to0}  \ad^{n}_{\ii B^\sigma_s} ( \varphi ( H_\alpha ) )$. Hence the statement of the lemma for $n$ is established, which concludes the proof.
\end{proof}

Lemma \ref{e22} shows that for all $\varphi \in \mathrm{C}_0^\infty( ( - \infty , e_{\mathrm{gap}} /2 ) ; \mathbb{R} )$, $\varphi(H_\alpha) \in \mathrm{C}^\infty( B^\sigma )$. We are now ready to prove the uniform bounds with respect to $\sigma$ on the commutators $\ad^{n}_{\ii B^\sigma} ( \varphi_\sigma ( H_\alpha - E_\alpha ) )$ given in Lemma \ref{c18}.

\begin{proof}[Proof of Lemma \ref{c18}]
We start as in the proof of Lemma \ref{e22} (see \eqref{e24}, \eqref{e25} and \eqref{e26}), considering $\varphi_\sigma(H_\alpha-E_\alpha) = \varphi( \sigma^{-1}( H_\alpha - E_\alpha ) )$ instead of $\varphi(H_\alpha)$, and introducing $\langle \sigma x \rangle$ instead of $\langle x \rangle$ everywhere. Whence the statement of the lemma will follow provided we estimate terms of the form
\begin{equation} \label{e35}
\int_{ \mathbb{R}^2 } \partial_{\bar z} \hat \varphi(z) \langle \sigma x \rangle^{m} \ad^{n}_{\ii B^\sigma} \big( \sigma^{-1} ( H_\alpha - E_\alpha ) - z )^{-1} \big) \langle \sigma x \rangle^{- ( n + m )} \, \d x \, \d y,
\end{equation}
uniformly in $\sigma$, for arbitrary $n,m \in \mathbb{N} \cup \{0\}$. By Lemma \ref{c19}, $\langle \sigma x \rangle^{m} \ad^{n}_{\ii B^\sigma} ( \sigma^{-1} ( H_\alpha - E_\alpha ) - z )^{-1} ) \langle \sigma x \rangle^{- ( n + m )}$ decomposes into a sum of terms of the form
\begin{align}
&\prod_{1 \leq l \leq n} \Big( \langle \sigma x \rangle^{t_{l-1}} \ad^{j_l}_{\ii B^\sigma} ( \sigma^{-1} H_\alpha ) \langle \sigma x \rangle^{- t_{l}} \big( \sigma^{-1} ( H_\alpha - E_\alpha )- z \big)^{-1}    \nonumber  \\
& \qquad \qquad \big( \sigma^{-1} ( H_\alpha - E_\alpha )- z \big) \langle \sigma x \rangle^{t_{l}} \big( \sigma^{-1} ( H_\alpha - E_\alpha )- z \big)^{-1} \langle \sigma x \rangle^{- t_{l}} \Big)   \notag \\
& \qquad \qquad \qquad \qquad \qquad \qquad \qquad \langle \sigma x \rangle^{n + m} \big( \sigma^{-1} ( H_\alpha - E_\alpha ) - z \big)^{-1} \langle \sigma x \rangle^{- ( n + m )} , \label{f1}
\end{align}
with $1\leq j_l \leq n$, $\sum_{l=1}^n j_l = n$, $t_0 = m$ and $t_l = m + \sum_{i=1}^l j_i$. From Lemma \ref{d7}, it follows
\begin{gather} \label{e27}
\big\Vert \big( \sigma^{-1} ( H_\alpha - E_\alpha )- z \big) \langle \sigma x \rangle^{t_{l}} \big( \sigma^{-1} ( H_\alpha - E_\alpha )- z \big)^{-1} \langle \sigma x \rangle^{- t_{l}} \big\Vert \leq \frac{\mathrm{C}_{t_l}}{| \im z |^{t_l}} ,  \\
\big\Vert \langle \sigma x \rangle^{n + m} \big( \sigma^{-1} ( H_\alpha - E_\alpha ) - z \big)^{-1} \langle \sigma x \rangle^{- ( n + m )} \big\Vert \leq \frac{\mathrm{C}}{| \im z |^{n + m + 1}} . \label{f2}
\end{gather}
It remains to estimate $\big\Vert \langle \sigma x \rangle^{t_{l-1}} \ad^{j_l}_{\ii B^\sigma} ( \sigma^{-1} H_\alpha ) \langle \sigma x \rangle^{-t_{l}} ( \sigma^{-1} ( H_\alpha - E_\alpha )- z )^{-1} \big\Vert $. To this end, we compute
\begin{align}
\langle \sigma x \rangle^{t_{l-1}}  \ad^{j_l}_{\ii B^\sigma} ( \sigma^{-1} H_\alpha ) \langle \sigma x \rangle^{- t_{l}} ={}& \sigma^{-1} \langle \sigma x \rangle^{-j_{l}} \d \Gamma \big( \eta_{\sigma}^{j_{l}} (k) |k| \big)     \notag \\
&+ (-1)^{j_l} \sigma^{-1} \langle \sigma x \rangle^{-j_l} \sum_{\fract{0 \leq p_1,p_2 \leq j_l}{p_1+p_2 = j_l}} \big( \Phi^{(p_1)} \cdot \Phi^{(p_2)} + \Phi^{(p_2)} \cdot \Phi^{(p_1)} \big) \notag \\
&+ 2\ii (-1)^{j_l+1} t_{l} \langle \sigma x \rangle^{- j_{l} - 1} \frac{\sigma x}{\langle \sigma x \rangle} \cdot \Phi^{(j_l)} . \label{e28}
\end{align}
where, recall that $\eta_{\sigma}^{j_l} (k) = \eta^{j_l} ( k / \sigma )$ with $\eta^{j_l} \in \mathrm{C}_0^\infty( \{ \vert k \vert \leq 1 \} )$ and that the $\Phi^{(j)}$'s are defined in \eqref{g9}--\eqref{g10}. For the first term in the right hand side of \eqref{e28}, we use Lemma \ref{e2} which implies
\begin{align}
\big\Vert \sigma^{-1} & \d \Gamma \big( \eta_{\sigma}^{j_{l}} (k) |k| \big) \big( \sigma^{-1} ( H_\alpha - E_\alpha )- z \big)^{-1} \big\Vert \notag \\
&\leq \big\Vert \sigma^{-1} ( \mathds{1}_{ \mathcal{H}_{\ge\sigma} } \otimes H_f ) \big( \sigma^{-1} ( H_\alpha - E_\alpha )- z \big)^{-1} \big\Vert \notag \\
&\leq \mathrm{C} \big\Vert \sigma^{-1} ( H_{\alpha} - E_{\alpha} ) \big( \sigma^{-1} ( H_\alpha - E_\alpha )- z \big)^{-1} \big\Vert + \mathrm{C} \big\Vert \big( \sigma^{-1} ( H_\alpha - E_\alpha )- z \big)^{-1} \big\Vert \notag \\
&\le \mathrm{C} + \frac{\mathrm{C}}{| \im z |} . \label{e29}
\end{align}
Next, using again Lemma \ref{d10}, one verifies that
\begin{equation} \label{e30}
\big\Vert \langle \sigma x \rangle^{-j_l} \Phi^{(p_1)} \cdot \Phi^{(p_2)} ( \mathds{1}_{ \mathcal{H}_{\ge\sigma}} \otimes H_f + \sigma )^{-1} \big\Vert \leq \mathrm{C} \sigma,
\end{equation}
for any $1 \leq p_1,p_2 \leq j_l$ such that $p_1+p_2=j_l$, that
\begin{equation} \label{e31}
\big\Vert \langle \sigma x \rangle^{-j_l} \Phi^{(j_l)} ( \mathds{1}_{ \mathcal{H}_{\ge\sigma}} \otimes H_f + \sigma )^{-\frac{1}{2}} \big\Vert \leq \mathrm{C} \sigma^{\frac{1}{2}} ,
\end{equation}
and that
\begin{align}
\big\Vert \langle \sigma x & \rangle^{-j_l} \Phi^{(j_l)} \cdot \Phi^{(0)} ( \mathds{1}_{ \mathcal{H}_{\ge\sigma}} \otimes H_f + \sigma )^{-\frac{1}{2}} ( H_{\alpha} - E_{\alpha} + 1 )^{-\frac{1}{2}} \big\Vert   \nonumber \\
\leq{}& \big\Vert \langle \sigma x \rangle^{-j_l} \Phi^{(j_l)} ( \mathds{1}_{ \mathcal{H}_{\ge\sigma}} \otimes H_f + \sigma )^{-\frac{1}{2}} \big\Vert \big\Vert \big( p + \alpha^{\frac{3}{2}} A_{\ge\sigma}( \alpha x ) \big) ( H_{\alpha} - E_{\alpha} + 1 )^{-\frac{1}{2}} \big\Vert \nonumber \\
&+ \big\Vert \langle \sigma x \rangle^{-j_l} \Phi^{(j_l)} \alpha^{\frac{3}{2}} A^{\le \sigma}( \alpha x ) ( \mathds{1}_{ \mathcal{H}_{\ge\sigma}} \otimes H_f + \sigma )^{- 1} \big\Vert \big\Vert ( \mathds{1}_{ \mathcal{H}_{\ge\sigma}} \otimes H_f + \sigma )^{\frac{1}{2}} ( H_{\alpha} - E_{\alpha} + 1 )^{-\frac{1}{2}} \big\Vert   \nonumber \\
\leq{}& \mathrm{C} \sigma^{\frac{1}{2}} ,   \label{e32}
\end{align}
since $\Phi^{(0)} = ( p + \alpha^{\frac{3}{2}} A_{\ge\sigma}( \alpha x ) ) + \alpha^{\frac{3}{2}} A^{\le \sigma}( \alpha x )$. Using Lemma \ref{e2}, it then follows from \eqref{e30} and \eqref{e32} that
\begin{equation} \label{e33}
\Big\Vert \langle \sigma x \rangle^{-j_l} \sum_{\fract{0 \leq p_1,p_2 \leq j_l}{p_1+p_2 = j_l}} \big( \Phi^{(p_1)} \cdot \Phi^{(p_2)} + \Phi^{(p_2)} \cdot \Phi^{(p_1)} \big) \big( \sigma^{-1} ( H_\alpha - E_\alpha )- z \big)^{-1} \Big\Vert \leq \frac{ \mathrm{C} \sigma}{ | \im z | },
\end{equation}
whereas \eqref{e31} implies that
\begin{equation} \label{e34}
\big\Vert \langle \sigma x \rangle^{- j_{l} - 1} \frac{\sigma x}{\langle \sigma x \rangle} \cdot \Phi^{(j_l)} \big( \sigma^{-1} ( H_\alpha - E_\alpha )- z \big)^{-1} \big\Vert \leq \frac{ \mathrm{C} \sigma }{ | \im z | }.
\end{equation}
Thus, combining \eqref{e28} with the estimates \eqref{e29}, \eqref{e33} and \eqref{e34}, we have shown
\begin{equation*}
\big\Vert \langle \sigma x \rangle^{t_{l-1}}  \ad^{j_l}_{\ii B^\sigma} ( \sigma^{-1} H_\alpha ) \langle \sigma x \rangle^{- t_{l}} \big( \sigma^{-1} ( H_\alpha - E_\alpha )- z \big)^{-1} \big\Vert \leq \frac{ \mathrm{C} }{ | \im z | },
\end{equation*}
which, combined with \eqref{f1}, \eqref{e27} and \eqref{f2}, leads to
\begin{equation*}
\big\Vert \langle \sigma x \rangle^{m} \ad^{n}_{\ii B^\sigma} \big( \sigma^{-1} ( H_\alpha - E_\alpha ) - z )^{-1} \big) \langle \sigma x \rangle^{- ( n + m )} \big\Vert \leq \frac{ \mathrm{C}_{n,m} }{ | \im z |^{ \gamma_{n,m} } } ,
\end{equation*}
where $\gamma_{n,m} := \sum_{l=1}^n t_l + 2 n + m + 1$. With \eqref{e35}, this concludes the proof of the lemma.
\end{proof}

\begin{remark}\sl \label{g8}
By similar (and simpler) arguments, one can also estimate the multiple commutators $\ad^n_{\ii B^\sigma}( \varphi(H_\alpha) )$ uniformly in $\sigma$. More precisely, for all $n \in \mathbb{N}\cup\{0\}$ and $\varphi \in \mathrm{C}_0^\infty( ( - \infty , E_\alpha + e_{\mathrm{gap}} / 2 ) ; \mathbb{R} )$, there exists $\mathrm{C}_{n,\varphi} > 0$ such that, for all $0 \leq \alpha \leq \alpha_c$ and $0 < \sigma \le e_{\mathrm{gap}} / 2$,
\begin{equation*}
\big\Vert \ad^n_{\ii B^\sigma} ( \varphi ( H_{\alpha} ) ) \big\Vert \leq \mathrm{C}_{n,\varphi} .
\end{equation*}
\end{remark}

The next lemma could be proven in the same way as Lemma \ref{c18}, using Lemma \ref{d7} with $\tau = \sigma$. The proof below is however much more simple, and simply follows from the commutation relation \eqref{g6}.

\begin{lemma}\sl \label{e36}
There exists $\alpha_c>0$ such that, for all $n \in \mathbb{N}\cup\{0\}$ and $\varphi \in \mathrm{C}_0^\infty( ( -\infty , 1 ) ; \mathbb{R} )$, there exists $\mathrm{C}_{n,\varphi} > 0$ such that, for all $0 \leq \alpha \leq \alpha_c$ and $0 < \sigma \leq e_{\mathrm{gap}}/2$,
\begin{equation*}
\big\Vert \ad^n_{\ii B^\sigma} ( \varphi_\sigma ( H_{\alpha,\sigma} - E_{\alpha,\sigma} ) ) \big\Vert \leq \mathrm{C}_{n,\varphi}.
\end{equation*}
\end{lemma}

\begin{proof}
A direct computation based on \eqref{g2}, \eqref{g6} and the Helffer-Sj{\"o}strand formula shows that the commutators $\ad^n_{\ii B^\sigma} ( \varphi_\sigma ( H_{\alpha,\sigma} - E_{\alpha,\sigma} ) )$ (defined iteratively in the sense of quadratic forms on $D(B^\sigma) \times D(B^\sigma)$) extend by continuity to bounded operators on $\mathcal{H}$, and that $\ad^n_{\ii B^\sigma} ( \varphi_\sigma ( H_{\alpha,\sigma} - E_{\alpha,\sigma} ) )$ decomposes into a sum of terms of the form
\begin{equation*}
\sigma^{- j} \Pi_{\alpha,\ge\sigma} \otimes \big( \d \Gamma \big( \eta_{\sigma}^{1} ( k ) |k| \big) \cdots \d \Gamma \big( \eta_{\sigma}^{j} ( k ) |k| \big) ( \varphi^{( j )} )_\sigma ( H_f ) \big),
\end{equation*}
with $j \in \N$ satisfying $1 \leq j \leq n$, $\eta_{\sigma}^{\#} (k) = \eta^{\#} ( k / \sigma )$ and $\eta^{\#} \in \mathrm{C}_0^\infty( \{ \vert k \vert \leq 1 \} )$. Using that $\d \Gamma \big( \eta_{\sigma}^{\#} ( k ) |k| \big)^{2} \leq H_f^{2}$, one easily obtains the required estimate.
\end{proof}

We conclude with the following proposition which was used in Section \ref{f4}.

\begin{proposition}\sl \label{e37}
There exists $\alpha_c>0$ such that, for all $\varphi \in \mathrm{C}_0^\infty( (-\infty, 1 ) ; \mathbb{R} )$, $n \in \mathbb{N}$ and $\delta > 0$, there exists $\mathrm{C}_{\varphi , n , \delta} > 0$ such that, for all $0 \leq \alpha \leq \alpha_c$ and $0 < \sigma \leq e_{\mathrm{gap}} / 2$,
\begin{equation*}
\big\Vert \ad_{\ii B^\sigma}^{n} \big( \varphi_\sigma( H_\alpha - E_\alpha ) - \varphi_\sigma( H_{\alpha,\sigma} - E_{\alpha,\sigma} ) \big) \big\Vert \leq \mathrm{C}_{\varphi , n , \delta} ( \alpha^{ \frac{3}{2} } \sigma )^{1-\delta}.
\end{equation*}
\end{proposition}

\begin{proof}
Let $\Phi \in \mathcal{H}$ be such that $\| \Phi \| = 1$ and let, for $s \in \mathbb{R}$,
\begin{equation*}
f(s) := \left \langle e^{\ii s B^\sigma} \Phi , \big( \varphi_\sigma(H_\alpha-E_\alpha) - \varphi_\sigma( H_{\alpha,\sigma} - E_{\alpha,\sigma} ) \big) e^{\ii s B^\sigma } \Phi \right \rangle.
\end{equation*}
It follows from Lemma \ref{c18}, Lemma \ref{e22} and Lemma \ref{e37} that 
\begin{equation*}
f^{(n)}(s) = \left \langle e^{\ii s B^\sigma} \Phi , \ad^n_{\ii B^\sigma} \big( \varphi_\sigma(H_\alpha-E_\alpha) - \varphi_\sigma( H_{\alpha,\sigma} - E_{\alpha,\sigma} ) \big) e^{\ii s B^\sigma } \Phi \right \rangle,
\end{equation*}
and that $ \| f^{(n)} \|_\infty \le \mathrm{C}_{\varphi,n}$ for all $n \in \mathbb{N}$. On the other hand, by Proposition \ref{c10}, we have $\| f \|_\infty \le \mathrm{C}_{\varphi} \alpha^{3/2} \sigma$. The Kolmogorov inequality then implies that
\begin{equation}
\big\| f^{(n)} \big\|_\infty \le \mathrm{C}_{n,m} \big\| f \big\|_\infty^{1-\frac{n}{m}} \big\| f^{(m)} \big\|_\infty^{ \frac{n}{m} }  \le \mathrm{C}_{\varphi,n,m} (\alpha^{\frac{3}{2}}\sigma)^{1-\frac{n}{m}}.
\end{equation}
for all $m \ge n$. Taking $m$ sufficiently large concludes the proof of the lemma.
\end{proof}

\begin{remark}\sl
Using a suitable Pauli-Fierz transformation as in the proof of \cite[Proposition 7]{FGS1}, one could presumably prove that
\begin{equation*}
\big\Vert \ad_{\ii B^\sigma}^{n} \big( \varphi_\sigma( H_\alpha - E_\alpha ) - \varphi_\sigma( H_{\alpha,\sigma} - E_{\alpha,\sigma} ) \big) \big\Vert \leq \mathrm{C}_{\varphi,n} \alpha^{ \frac{3}{2} } \sigma,
\end{equation*}
for all $n \in \mathbb{N} \cup \{0\}$. For the purpose of the present paper, however, the statement of Proposition \ref{e37} is sufficient.
\end{remark}

\bibliographystyle{amsalpha}

\end{document}